\documentclass[a4paper,UKenglish,cleveref, autoref, thm-restate]{lipics-v2021}

\usepackage{todonotes}
\usepackage{microtype}

\title{Constrained Flips in Plane Spanning Trees}

\author{Oswin Aichholzer}{Graz University of Technology, Austria}{oswin.aichholzer@tugraz.at}{https://orcid.org/0000-0002-2364-0583}{}
\author{Joseph Dorfer}{Graz University of Technology, Austria}{joseph.dorfer@tugraz.at}{https://orcid.org/0009-0004-9276-7870}{Austrian Science Foundation (FWF) 10.55776/DOC183}
\author{Birgit Vogtenhuber}{Graz University of Technology, Austria}{birgit.vogtenhuber@tugraz.at}{https://orcid.org/0000-0002-7166-4467}{}

\authorrunning{O. Aichholzer and J. Dorfer and B. Vogtenhuber}

\Copyright{}



\nolinenumbers

\graphicspath{{./figures/}}

\hideLIPIcs
\begin{document}
	
	\maketitle
	
\begin{abstract}
A flip in a plane spanning tree $T$ is the operation of removing one edge from $T$ and adding another edge
such that the resulting structure is again a plane spanning tree. For trees on a set of points in convex position we study two classic types of constrained flips:
(1)~Compatible flips are flips in which the removed and inserted edge do not cross each other.  
We relevantly improve the previous upper bound of $2n-O(\sqrt{n})$ on the diameter of the compatible flip graph to~$\frac{5n}{3}-O(1)$, 
by this matching the upper bound for unrestricted flips by Bjerkevik, Kleist, Ueckerdt, and Vogtenhuber [SODA~2025] up to an additive constant of $1$. We further show that no shortest compatible flip sequence removes an edge that is already in its target position. Using this so-called happy edge property, we derive a fixed-parameter tractable algorithm to compute the shortest compatible flip sequence between two given trees.
(2)~Rotations are flips in which the removed and inserted edge share a common vertex.
Besides showing that the happy edge property does not hold for rotations, we improve the previous upper bound of $2n-O(1)$ for the diameter of the rotation graph to~$\frac{7n}{4}-O(1)$.
	\end{abstract}

	\section{Introduction}\label{sec:intro}
	
	Let $S$ be a finite point set in the plane in general position, that is, no three points lie on a common line. We call $S$ a \emph{convex} point set if no point in $S$ lies in the interior of the convex hull of $S$.
	A plane straight-line graph of $S$ is a graph with vertex set $S$ and whose edges are straight line segments between pairs of points of S such that no two edges intersect, except at a common endpoint. All graphs considered in this paper are straight-line graphs. For brevity we will omit the term straight-line. A \emph{plane spanning tree} is a plane connected graph with~$n-1$ edges. In a convex point set, edges appear in two forms. \emph{Convex hull edges} are edges that lie on the boundary of the convex hull of the point set. Edges that are not convex hull edges are \emph{diagonals}.
	
	\subparagraph*{Flip Graphs of Plane Spanning Trees.}
	A \emph{flip} between two plane spanning trees of $S$ is the operation that removes an edge from a tree and adds another edge such that the resulting structure is again a plane spanning tree. We also denote this operation as an \emph{unrestricted flip}. 
	A constrained version of flips are so-called \emph{compatible flips}, for which the removed edge and the added edge are not allowed to cross. \emph{Rotations} are a restricted subclass of compatible flips for which we additionally require the removed edge and the added edge to share a common vertex. The most restricted and local operation is a so called \emph{slide}. A slide is an operation in which the added edge $(u,v)$ and the removed edge $(v,w)$ are again required to share a common vertex, additionally, the triangle $\Delta(u,v,w)$ must not contain any other vertex and the edge $(u,w)$ has to be contained in the current tree. This can visually be interpreted as sliding the edge $(u,v)$ along the edge $(u,w)$ into the edge $(v,w)$.
	The \emph{(compatible) flip graph} of plane spanning trees of $S$ has as its vertex set all such trees. 
	Two vertices $T_1$, $T_2$ of this flip graph are connected with an edge if and only if $T_1$ can be transformed into $T_2$
	by a single (compatible) flip. See Figure~\ref{flips} for an example of flips in plane spanning trees.
	
	\begin{figure}[ht]
		\centering
		\includegraphics[scale=0.6]{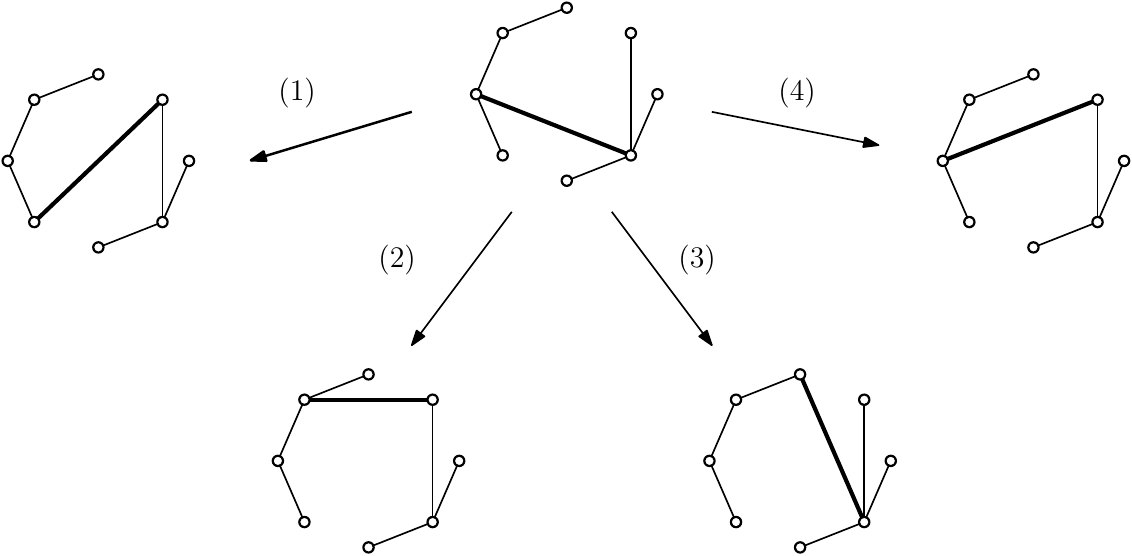}
		\caption{Four flips in a plane spanning tree. (1) unrestricted flip, (2) compatible flip, (3) rotation, (4) slide}
		\label{flips}
	\end{figure}
	
	Given an initial tree $T_{\text{in}}$ and a target tree $T_{\text{tar}}$, a \emph{(compatible) flip sequence} from $T_{\text{in}}$ to~$T_{\text{tar}}$ is a path from $T_{\text{in}}$ to~$T_{\text{tar}}$ in the (compatible) flip graph. The \emph{(compatible) flip distance} between $T_{\text{in}}$ and $T_{\text{tar}}$ is the length of a shortest path between $T_{\text{in}}$ and $T_{\text{tar}}$ in the flip graph. 
	Any (compatible) flip sequence of this length is called a \emph{shortest (compatible) flip sequence}. Similarly, we define the \emph{rotation graph}, \emph{rotation sequences} and the \emph{rotation distance}. See Figure \ref{flip_graph} for an illustration of the concepts.
	
	\begin{figure}[ht]
		\centering
		\includegraphics[scale=0.6]{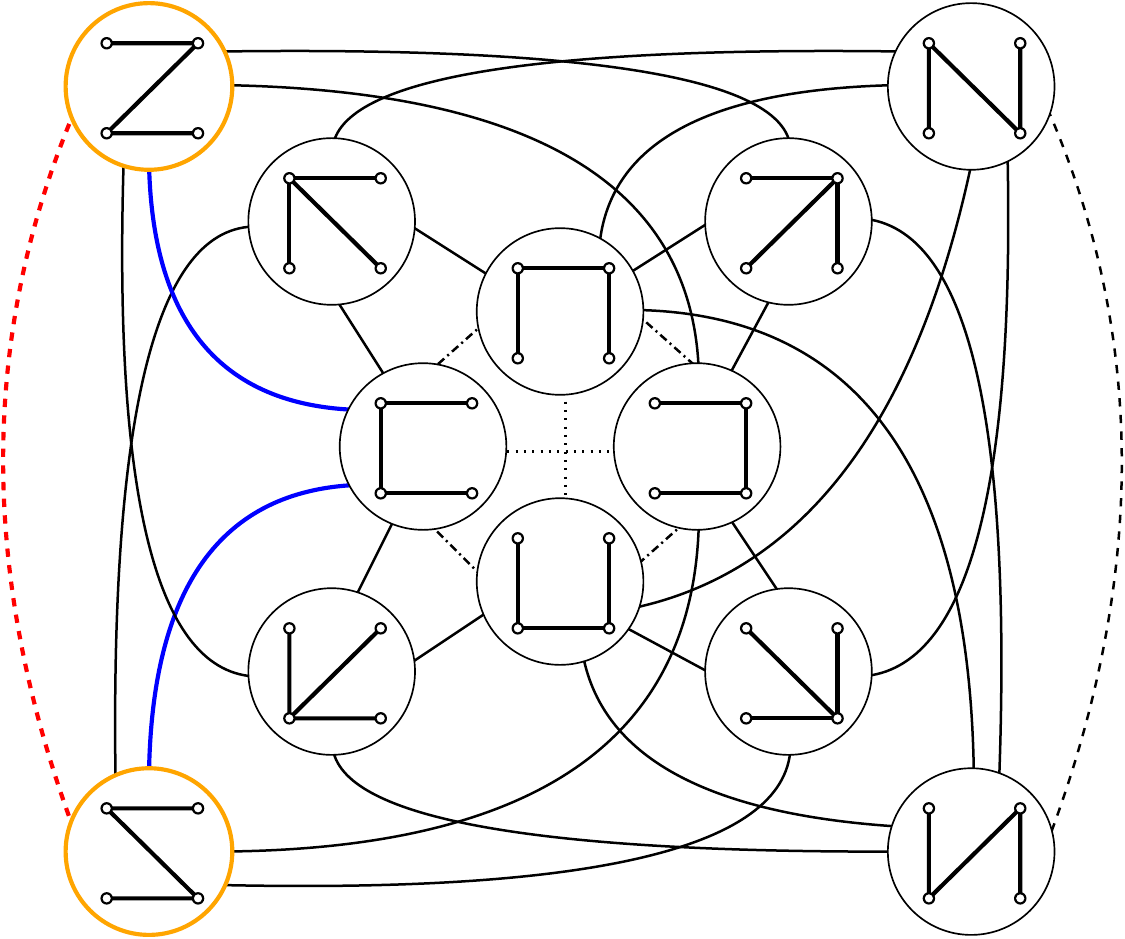}
		\caption{Flip graph of plane spanning trees on four vertices in convex position. Dashed edges denote flips that are part of the flip graph, but not the compatible flip graph. The dotted lines are compatible flips, but not rotations and the dash dotted lines are rotations but not slides. One pair of two plane spanning trees is marked in orange. Their shortest flip sequence of length one is marked in red (thicker dashed). A shortest compatible flip sequence of length two is marked in blue (thicker).}
		\label{flip_graph}
	\end{figure}
	
	For graph reconfiguration problems three natural questions arise: (1) Is the flip graph connected? (2) What is the diameter of the flip graph? (3) What is the complexity of computing shortest flip sequences between two specific configurations?
	
	\subparagraph*{Connectedness and Diameter.} For any $n$-point set $S$, the flip graph of plane spanning trees on $S$
	is known to be connected and has radius exactly $n-2$.
	The same holds for the compatible flip graph and the rotation graph; see 	~\cite{AVIS199621,bjerkevik2024flippingnoncrossingspanningtrees}. 
	Hence, the diameter always lies between $n-2$ and $2n-4$. 
	Since a lower bound for the diameter of $\big\lfloor\frac{3}{2}n\big\rfloor-5$ for convex $n$-point sets was shown in~\cite{hhmmn-gtgpcp-99}, the flip graph 
	of plane spanning trees on convex point sets has received considerable attention.  
	However, despite a prevailing conjecture that the diameter is at most $\frac{3}{2}n+O(1)$~\cite{bousquet2023notes,ferran03CCCGpersonalcommunication}, even improving the upper bound to $2n-o(n)$ took surprisingly long.
	In the last few years, the upper bound for its diameter was improved to $2n - \Omega(log(n))$ in~\cite{aichholzer2024reconfiguration} and soon after to $2n-\Omega(\sqrt{n})$ in~\cite{bousquet2023notes}. 
	The upper bound from~\cite{bousquet2023notes} is constructive and, though not stated explicitly, only uses compatible flips.
	Hence it also provides an upper bound for the diameter of the according compatible flip graph. 
	In~\cite{bousquet_et_al:LIPIcs.SoCG.2024.22}, the diameter was bounded from above by~$cn$ with $c = \frac{1}{12}\left(22+\sqrt{2}\right) \approx 1.95$, which marked the first linear improvement over the initial bound from~\cite{AVIS199621}. 
	Very recently, a lower bound of $\frac{14}{9}n-O(1)$ and an upper bound of $\frac{5}{3}n-3$ on the diameter were achieved in~\cite{bjerkevik2024flippingnoncrossingspanningtrees}, 
	where the upper bound in general requires non-compatible~flips.
	It is also known that any plane spanning tree can be transformed into any other using $\Theta(n^2)$ many slides \cite{AICHHOLZER2007155}. For a survey on further flip operations in trees with multiple flips in parallel as well as flips with labeled edges, we refer~to~\cite{NICHOLS2020111929}.
	
	\subparagraph*{Happy Edges and Complexity.} For any graph reconfiguration problem, \emph{happy edges} are edges that lie in both the initial and the target graph.
	A flip graph fulfills the \emph{happy edge property} if there exists a shortest flip sequence between any two graphs that never flips happy edges.
	The happy edge property often is a good indication for the complexity of a graph reconfiguration problem. 
	For example, for triangulations of simple polygons~\cite{Aichholzer2015} and general point sets~\cite{LUBIW201517,Pilz_2014},
	finding shortest flip sequences is computationally hard. The hardness proofs
	use counterexamples to the happy edge property as a key ingredient. 
	Conversely, the happy edge property is known to hold for triangulations of convex polygons~\cite{2ae353801697435f901a750248959ba2}.
	Though the complexity of finding shortest flip sequences is still open, the happy edge property yields multiple fixed-parameter tractable algorithms~\cite{articlefpt2,cleary2009rotation,li20233,articlefpt}.
	Moreover,
	the happy edge property holds for plane perfect matchings of convex point sets and shortest flip sequences can be found in polynomial time~\cite{articlematchings}. For plane spanning trees in convex position the happy edge property was conjectured to hold~\cite{aichholzer2024reconfiguration,ferran03CCCGpersonalcommunication} for different types of flips. Up to now, it has only been disproved for slides~\cite{aichholzer2024reconfiguration}.

	\subparagraph*{Contribution and Outline.} We will start by improving upper bounds on the flip graph for restricted flip types in convex point sets. In Section \ref{sec:upperbound} we improve the upper bound on the diameter of the compatible flip graph from $2n-\sqrt{n}$ in~\cite{bousquet2023notes} to $\frac{5}{3}(n-1)$. Section \ref{sec:rotations} is dedicated to improving the upper bound on the rotation graph from $2n-O(1)$ in \cite{AVIS199621} to $\frac{7}{4}(n-1)$.
	
	Afterwards we deal with the happy edge property. In Section \ref{sec:happyedges} we provide a more refined framework for happy edge properties and in Section \ref{sec:happycompatible} we prove the happy edge property for compatible flips in convex point sets and provide a counterexample to the happy edge property for rotations.
	
	Finally, in Section \ref{sec:fpt} we use variants of the happy edge property to show fixed-parameter tractability of the flip distance for both unrestricted and compatible flips.
	
	\emph{Due to space restrictions many details have been moved to the appendix. Such results will be marked with $(\star)$.}
	
	\section{Bounding the Diameter of the Compatible Flip Graph}\label{sec:upperbound}
	This section is devoted to improving the upper bound on the diameter of the compatible flip graph.
	All the currently known improvements to the upper bound of the length of shortest flip sequences for unrestricted flips are
	based on two different ways to handle edges:
	They either invest two flips to flip an edge from $T_{\text{in}}\setminus T_{\text{tar}}$ to a convex hull edge and later flip that convex hull edge to an edge of $T_{\text{tar}}\setminus T_{\text{in}}$, or perform a single \emph{perfect flip}, that is, a flip from an edge of $T_{\text{in}}\setminus T_{\text{tar}}$ to an edge of $T_{\text{tar}}\setminus T_{\text{in}}$. 
	Upper bounding the length of a flip sequence is then achieved by lower bounding the number of perfect flips.
	Note that any such flip sequence induces a pairing of the edges of $T_{\text{in}}\setminus T_{\text{tar}}$ with edges of  $T_{\text{tar}}\setminus T_{\text{in}}$.
	
	To obtain our improved bound for the diameter of the compatible flip graph, 
	we build on strategies for bounding the diameter of the flip graph with unrestricted flips from~\cite{bjerkevik2024flippingnoncrossingspanningtrees}. We give an overview of their proof and highlight the parts where we provide adjustments to make the flip sequence compatible.
	
	The pairing of edges in~\cite{bjerkevik2024flippingnoncrossingspanningtrees} is obtained as follows: The authors consider the linear representation of a tree $T$ on a convex point set as illustrated in Figure~\ref{fig:linear}. The linear representation is obtained by cutting the convex point set between two consecutive vertices on the convex hull and unfolding the tree $T$ such that its vertices are on a straight line. 
	
	\begin{figure}[ht]
		\centering
		\includegraphics[scale=0.6]{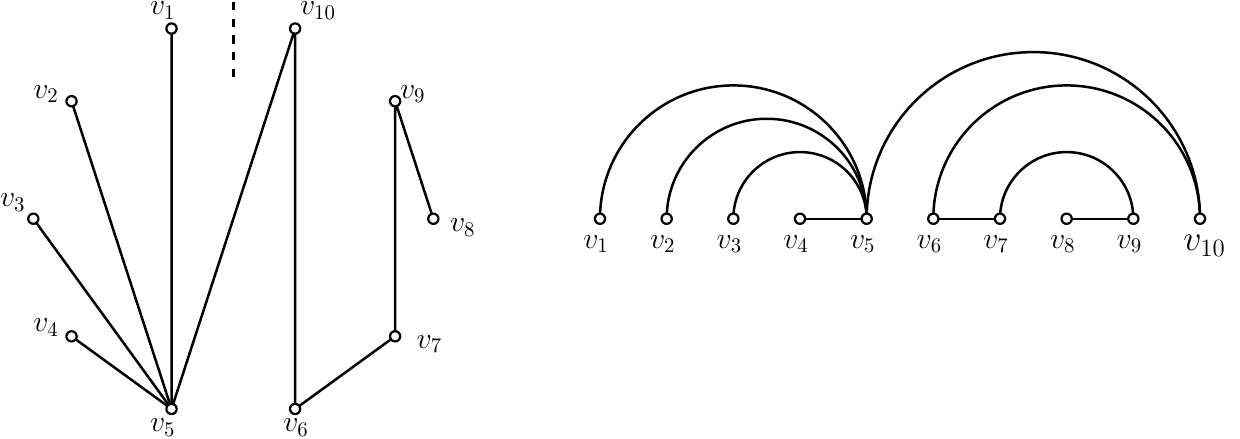}
		\caption{A tree $T$ on a convex point set (left) and a linear representation of $T$ (right)} 
		\label{fig:linear}
	\end{figure}
	
	The linear representation induces a linear order $v_1,\ldots,v_n$ of the vertices of $T$.
	With this linear order, the \emph{length of an edge} $e=(v_i,v_j)$, where $i<j$, is defined as $j-i$. 
	Further, the edge $e$ \emph{covers} all vertices $\{v_i, \ldots, v_j\}$ as well as any edge $(v_k,v_\ell)$ with $k,\ell \in \{i, \ldots,j\}$.

	A \emph{gap} $\{v_i,v_{i+1}\}$ is an edge between two consecutive vertices in the linear order (which might or might not be present in the tree $T$).
	Next, the authors define a bijection (Lemma~3.1 in~\cite{bjerkevik2024flippingnoncrossingspanningtrees}) between \emph{gaps} and edges of $T$,
	where each gap is in bijection with the shortest edge of~$T$ that covers that gap. 
	The edges are then partitioned into 
	categories based on how many vertices they share with their corresponding gap $(v_i,v_{i+1})$; see Figure \ref{fig:snw} for an illustration.
	
	\pagebreak
	
	Concretely, an edge $(u,v)$ is a \dots
	\begin{itemize}
		\item \emph{short edge} if $\{u,v\} = \{v_i,v_{i+1}\}$
		\item \emph{near edge} if $\lvert\{u,v\}\cap\{v_i,v_{i+1}\}\rvert = 1$
		\item \emph{wide edge} if $\lvert\{u,v\}\cap\{v_i,v_{i+1}\} \rvert = 0$
	\end{itemize}
	
	\begin{figure}[ht]
		\centering
		\includegraphics[scale=0.6]{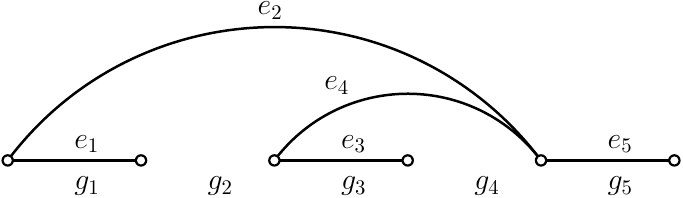}
		\caption{Edges $e_1$, $e_3$ and $e_5$ are short edges for gaps $g_1$, $g_3$ and $g_5$ respectively, $e_4$ is a near edge for~$g_4$, and~$e_2$ is a wide edge for $g_2$.} 
		\label{fig:snw}
	\end{figure}
	
	The authors show that the number of short edges in any tree $T$ is always larger than the number of wide edges. 
	In a next step, they construct an edge-edge-bijection between the edges in $T_{\text{in}}$ and $T_{\text{tar}}$
	by concatenating the edge-gap-bijection 
	for $T_{\text{in}}$ with the gap-edge-bijection of $T_{\text{tar}}$.
	An edge $e \in T_{\text{in}}$ and its image from the edge-edge-bijection $e'\in T_{\text{tar}}$
	form a \emph{pair}. If an edge $e$ is in both $T_{\text{in}}$ and $T_{\text{tar}}$, and $e$ is in bijection to itself, then $e$ is not flipped at all.
	Note that any edge can be flipped to its assigned gap by a single (compatible) flip. 
	The authors use this and the relation between the numbers of short and wide edges in each tree  
	to show that all pairs of edges containing at least one short or wide edge can be flipped using an average of $\frac{3}{2}$ flips per pair. 
	This is done by flipping every edge in such a pair to its corresponding gap and later from its gap to its target edge,
	which takes two flips minus the number of short edges in the pair.
	
	The remaining pairs then only consist of near edges. These pairs can further be partitioned into three sets as illustrated in Figure \ref{fig:abc}.
	\begin{itemize}
		\item \emph{above pairs}, where $e$ and $e'$ share a vertex and $e$ is longer than $e'$.
		\item \emph{below pairs}, where $e$ and $e'$ share a vertex and $e$ is shorter than $e'$.
		\item \emph{crossing pairs}, where $e$ and $e'$ cross, or equivalently, $e$ and $e'$ share a different vertex with their gap.
	\end{itemize}
	
	\begin{figure}[ht]
		\centering
		\includegraphics[scale=0.6]{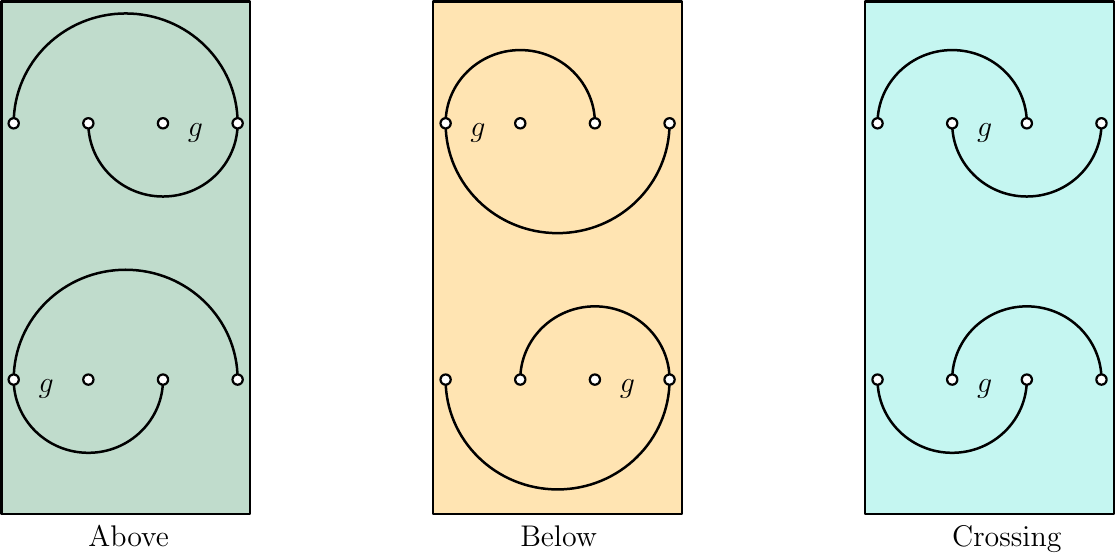}
		\caption{Illustration of above, below and crossing pairs of edges. Edges of $T_{\text{in}}$ are drawn above the vertices, while edges of $T_{\text{tar}}$ are drawn below the vertices.
		The respective gap is denoted by $g$.} 
		\label{fig:abc}
	\end{figure}
	
	The sets of above, below and crossing pairs are then denoted by $A$, $B$, and $C$, respectively. 
	The authors provide a flip sequence that flips all pairs in $A$ (resp.\ $B$, $C$) directly~\cite[Lemma~3.2]{bjerkevik2024flippingnoncrossingspanningtrees}. 
	Since a direct flip of an above or below pair is compatible, this implies the following lemma.
	
	\begin{restatable}{lemma}{lemmaCompatibleAB}\label{Lemma2}
		Assume that in the linear representation, the trees $T_{\text{in}}$ and $T_{\text{tar}}$ consist only of short edges and
		pairs of edges
		in $A$ (resp.\ $B$). 
		Then there is a compatible flip sequence from $T_{\text{in}}$ to $T_{\text{tar}}$ of length $\lvert A \rvert$ (resp.\ $\lvert B\rvert$).
	\end{restatable}
	
	In contrast, no direct flip between two edges of a pair in $C$ is compatible, since any such pair is crossing. 
	In Appendix \ref{app:upper} we develop an approach to flip these edges 
	such that the resulting flip sequence is compatible while using only one additional flip.
	
	\begin{restatable}[$\star$]{lemma}{lemmaCompatibleC}\label{Lemma3}
		Assume that in the linear representation, the trees $T_{\text{in}}$ and $T_{\text{tar}}$ consist only of short edges and 
		pairs of edges in $C$. 
		Then there is a compatible flip sequence from $T_{\text{in}}$ to $T_{\text{tar}}$ of length~$\lvert C \rvert +1$.
	\end{restatable}
	
	To obtain their upper bound on the flip distance between two trees $T_{\text{in}}$, $T_{\text{tar}}$, 
	the authors cut the original instances $T_{\text{in}}$, $T_{\text{tar}}$ along happy diagonals, handle each of the resulting subinstances separately, 
	and then add up the flips needed in each of the subinstances.
	Let $b$ denote the number of happy convex hull edges, $c$~the number of happy diagonals, and $d$ the cardinality of $T_{\text{in}}\setminus T_{\text{tar}}$. 
	Using that $b+c+d=n-1$, 
	their resulting bound~\cite[Theorem~4.1]{bjerkevik2024flippingnoncrossingspanningtrees}  
	is $\frac{5}{3}d+\frac{2}{3}b-\frac{4}{3} \leq \frac{5}{3}(n-1) - \frac{4}{3}$. 
	With similar calculations, we show the following theorem for compatible flips in Appendix~\ref{app:upper}.
	
	\begin{restatable}[$\star$]{theorem}{theoremCompatibleUpperBound}\label{upper}
		Let $T_{\text{in}}$ and $T_{\text{tar}}$ be two plane spanning trees in convex point sets with~$n$ vertices. 
		Then, there exists a compatible flip sequence from $T_{\text{in}}$ to $T_{\text{tar}}$ of length at most $\frac{5}{3}d + \frac{2}{3}b + c-\frac{1}{3} \leq \frac{5}{3}(n-1)-\frac{1}{3}$, %
		where $d=\lvert T_{\text{in}}\setminus T_{\text{tar}}\rvert=\frac{\lvert T_{\text{in}}\Delta T_{\text{tar}}\rvert}{2}$ and where $b$ and $c$ denote the number of happy convex hull edges and happy non-convex-hull edges, respectively.
	\end{restatable}
	
	We remark that our bound with respect to the parameters $d$, $b$, and $c$ can be relevantly higher than the one in~\cite[Theorem~4.1]{bjerkevik2024flippingnoncrossingspanningtrees} if the two trees share many happy diagonals. 
	This behaviour
	was not unexpected due to the following result: 
	In~\cite[Theorem~1.6]{bousquet2023reconfiguration}, the authors provide examples of pairs of plane spanning trees with (unrestricted) flip distance $d$ and compatible flip distance $2d$.
	The according tree pairs have
	$c\approx \frac{n}{2}$ happy diagonals.
	However, any pair of trees that could reach the upper bounds in~\cite[Theorem~4.1]{bjerkevik2024flippingnoncrossingspanningtrees} and Theorem~\ref{upper} must have $c=0$. So while for certain instances the flip distance and the compatible flip distance may differ strongly, extremal examples that determine the diameter are not affected.

\section{Bounding the Diameter of the Rotation Graph}\label{sec:rotations}
Next, we show that for the even more constrained rotation as flip operation, the diameter of the resulting flip graph is linear with leading coefficient less than two. %

For this purpose we again construct a pairing of edges such that a perfect flip can be performed for many of them. The challenge that arises, compared to unrestricted or compatible flips, is that rotations require the removed and the added edge to share a vertex. This rules out the gap based pairing from the last section, since in crossing pairs, edges do not share vertices and the same holds for most pairs that contain a wide edge.

\begin{lemma}\label{lem:bij_ve}
	For a plane spanning tree $T$ on a point set $S$ and a vertex $v \in S$ there exists a bijection $\rho_{T,v}$ between vertices in $S\setminus\{v\}$ and edges in $T$ such that $\rho_{T,v}(w)$ is incident to $w$ for every vertex $w \in S\setminus\{v\}$.
\end{lemma}

\begin{proof}
	We root the tree $T$ at $v$ and orient every edge away from $v$ such that $v$ forms the unique source of the oriented graph; see  Figure \ref{concepts} for an illustration. 
	Now map each vertex in $S\setminus\{v\}$ to its incoming edge.
\end{proof}

\begin{corollary} \label{bij_ee}
	Let $T_1$ and $T_2$ be two plane spanning trees on a point set $S$ and let $v$ be a vertex in $S$. 
	Then there exists a pairing of edges from $T_1$ with edges in $T_2$ such that:
	\begin{itemize}
		\item Two edges that are paired together share a vertex.
		\item Different pairs of edges share a different vertex.
		\item No pair of edges has $v$ as its shared vertex.
	\end{itemize}
\end{corollary}

\begin{proof}
	The pairing $\{(\rho_{T_1,v}(w),\rho_{T_2,v}(w)\mid w \in S \setminus \{v\}\}$ induced by the bijection in Lemma \ref{lem:bij_ve} has all the desired properties.
\end{proof}

\begin{figure}[ht]
	\centering
	\includegraphics[page=2, scale=0.6]{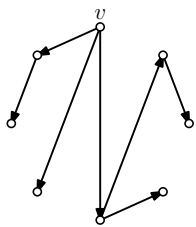}
	\caption{Illustration of the bijection from Lemma~\ref{lem:bij_ve}. 
		All edges are oriented away from the root $v$. 
		Each vertex in $S\setminus\{v\}$ is in bijection with its unique incoming edge.} 
	\label{concepts}
\end{figure}

For any two given plane spanning trees $T_{\text{in}}$ and $T_{\text{tar}}$, we apply the bijection from Corollary~\ref{bij_ee} with $v:=v_n$ to obtain a pairing of the edges of the two trees with all the mentioned properties.   
For an edge $e$ that is in bijection with a vertex $v_i$ we say that $e$ is \emph{attached} to $v_i$.
Note that each vertex $v_i$ with $i \in\{1,\ldots, n-1\}$ has an attached edge from each of $T_{\text{in}}$ and $T_{\text{tar}}$, while $v_n$ has no attached vertices.

From this point on we consider the two trees presented in their linear representation, see Figure \ref{fig:linear}, such that $v_{n}$ is the rightmost vertex. 
The pairs of edges $(e_{\text{in}}=(v_i,v_j),e_{\text{tar}}=(v_i,v_k))$, where $v_i$ is the vertex where edges $e_{\text{in}}$ and $e_{\text{tar}}$ are attached to, can now be subdivided into multiple classes. We say that a pair of edges is \dots
\begin{itemize}
	\item \emph{right-attached above} if $j\leq k<i$
	\item \emph{left-attached above} if $i<k\leq j$
	\item \emph{right-attached below} if $k \leq j<i$
	\item \emph{left-attached below} $i<j \leq k$
	\item \emph{diving} if $j<i<k$
	\item \emph {jumping} if $k<i<j$
\end{itemize}

We denote the respective sets of pairs in that order by $RA$, $LA$, $RB$ $LB$, $D$ and $J$. 
Additionally, we set $L = LA\cup LB$ and $R = RA \cup RB$. 
See Figure~\ref{pairs} for an illustration of the six types of edge pairs, as well as the two groups $L$ and $R$.
We remark that happy edges that are left-attached fulfill both, the definition for left-attached above and the one left-attached below. 
This does not pose any problem, since in the parts of the proof where this is relevant, we will only consider the whole set $L$. 
By symmetry, the same can be said about right-attached happy edges. 
Another noteworthy observation is that we do not handle convex hull edges any differently than diagonals. 
There might even be happy convex hull edges that belong to different pairs. 
This will lead to cases in which we rotate happy (convex-hull) edges. 
We allow this on purpose, since we will see in Section \ref{sec:happycompatible} that the happy edge property does not hold for rotations.

\begin{figure}[ht]
	\centering
	\includegraphics[scale=0.6]{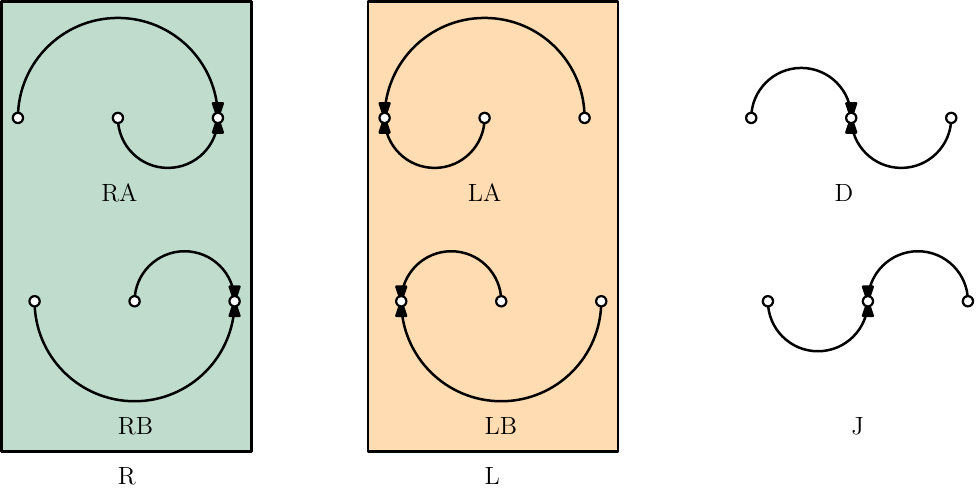}
	\caption{Illustration of pairs of edges in the sets RA, RB, R, LA, LB, L, D and J. 
	Edges of $T_{\text{in}}$ are drawn above the vertices, while edges of $T_{\text{tar}}$ are drawn below the vertices. Each edge points to the vertex which it is attached to.}
	\label{pairs}
\end{figure}

The new upper bound on the diameter of the rotation graph is described by the following Theorem.

\begin{theorem} \label{prediam}
	There exists a rotation sequence from $T_{\text{in}}$ to $T_{\text{tar}}$ of length at most\\$2(n-1) -  max\big\{\lvert L \rvert, \lvert R \rvert, \lvert J \rvert, \lvert D \rvert\big\} \leq \frac{7}{4}(n-1)$
\end{theorem}

We will use Lemma~\ref{convex_hull_star} below to get the edges from three of the sets $L$, $R$, $J$, and $D$ ``out of the way''. Then we will show that all edges from the remaining set can be perfectly flipped. 

We pair the edges and vertices of $T$ via the bijection $\rho_{T,v_{n}}$. 
Further, for an arbitrary $j\in\{1,\ldots, n-1\}$, we assign a gap to every vertex except $v_n$ by the following function (indices taken modulo $n$). 
\begin{equation*}
	\delta_j:	v_k \mapsto \begin{cases}
		(v_k,v_{k+1}) \qquad \text{if}~k> j\\
		(v_{k-1},v_k) \qquad \text{if}~k\leq j
	\end{cases}
\end{equation*}

In other words, we choose one particular gap $(v_j,v_{j+1})$ that is not assigned to any vertex and the other gaps are assigned such that they contain their assigned vertex.

In the following lemma, $j$ is an arbitrary index, $I$ is a subset of the convex hull edges in~$T$ that behave ``nicely'' with respect to $\delta_j$, 
and $K$ are all edges of $T$ that are incident to~$v_n$. $I^\ast$ and $K^\ast$ are according subsets that are required to be present in the resulting tree $T^\ast$.

\begin{lemma}\label{convex_hull_star}
	For any $j \in \{1, \ldots, n-1\}$, 
	let $I = \{k \in \{1, \ldots, n-1\} \mid \rho_{T,v_n}(v_k) = \delta_j(v_k)\}$ and $K = \{k \in~\{1, \ldots, n-1\} \mid \rho_{T,v_n}(v_k) = (v_n,v_k)\}$. 
	Further, let $I^\ast, K^\ast$ be a partition of $\{1,...,n-1\}$, that is, $I^\ast \cap K^\ast = \emptyset$ and $I^\ast \cup K^\ast = \{1, \ldots, n-1\}$. 
	Then there exists a rotation sequence from $T$ to $T^\ast$ of length at most $n-1-\lvert I \cap I^\ast\rvert - \lvert K \cap K^\ast\rvert$ that does not flip edges in $\rho_{T,v_n}(\{v_k\mid k \in(I \cap I^\ast)\cup(K\cap K^\ast)\})$ and $T^\ast$ has the following properties:
	\begin{itemize}
		\item For all $k\in I^\ast$, the tree $T^\ast$ contains the convex hull edges $\delta_j(v_k)$. That is, the convex hull edges attached to $v_k$ directed away from $v_n$ with respect to the order of vertices along the path on the convex hull from $v_n$ to $(v_j,v_{j+1})$\footnote{We remark that this direction in general differs from the orientation induced by the tree that is used in the proof of Lemma~\ref{lem:bij_ve}.}.
		\item The tree $T^\ast$ contains the edges $(v_n,v_k)$ for all $k \in K^\ast$. That is, all vertices that are attached to $v_k$ are connected directly to $v_n$ via edges that form a fan at $v_n$.
	\end{itemize}
\end{lemma}

We refer to Figure \ref{fig:concepts} for an illustration of Lemma~\ref{convex_hull_star} applied to the tree from Figure~\ref{concepts}.

\begin{proof}
	We set $E_0 = \rho_{T,v_n}(\{v_k\mid k \in(I \cap I^\ast)\cup(K\cap K^\ast)\})$ and $T_0 = T \setminus E_0$. Now consider the forest~$T_{0,vis}$ of edges that are \emph{visible} from $v_n$ in $T_0$. These are all edges for which no straight-line segment from $v_n$ to the edge is intersected by any other edge. 
	Further, let $v_k$ be a vertex of degree one in~$T_{0,vis}$. Then the only edge incident to $v_k$ in $T_{0,vis}$ is $\rho_{T,v_n}(v_k)$. 
In particular, if $k \in I^\ast$ then $\delta_j(v_k)$ does not contain any edge and is visible from $v_n$, 
and if $k \in K^\ast$ then the straight-line segment between $v_n$ and $v_k$ is uncrossed.
	
	We can execute a flip that exchanges the edge $\rho_{T,v_n}(v_k)$ with the edge $\delta_j(v_k)$ (if $k \in I^\ast$) or $(v_n,v_k)$ if $k\in K^\ast$ to get a valid tree.
	We replace $E_0\leftarrow E_0 \cup \{\delta_j(v_k)\}$ (resp.\ $E_0\leftarrow E_0 \cup \{(v_n,v_k)\}$) and $T_0\leftarrow T_0\setminus\{\delta_j(v_k)\}$ (resp.\ $T_0\leftarrow T_0 \setminus \{(v_n,v_k)\}$) and repeat the process. We repeat until $E_0=T$. This takes at most $n-1-\lvert I \cap I^\ast\rvert - \lvert K \cap K^\ast\rvert$ rotations, since the cardinality of $E_0$ increases by one in every iteration.
\end{proof}

\begin{figure}[ht]
	\centering
	\includegraphics[page=1, scale=0.55]{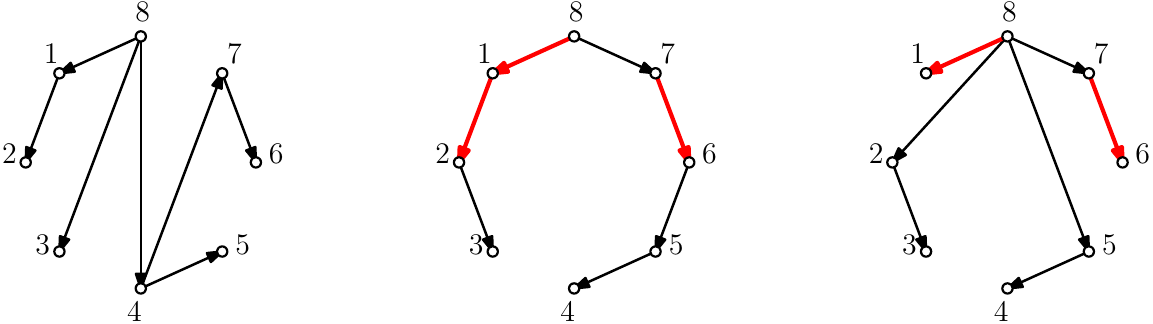}
	\caption{Illustration of Lemma \ref{convex_hull_star}. 
	Left: The tree from Figure~\ref{concepts} with $n = 8$. 
	Middle: Target structure of Lemma~\ref{convex_hull_star} with~$K^\ast=\emptyset$ and chosen index $j=3$. %
		The edges in $\rho_{T,v_n}(I \cap I^\ast) = \{(v_8,v_1),(v_1,v_2),(v_7,v_6)\}$ that do not get rotated when rotating from the left to the middle tree are marked in red. 
	Right: Target structure of Lemma~\ref{convex_hull_star} with $K^\ast =\{2,5\}$ and the same index $j=3$ as before. 
		The edges in $\rho_{T,v_n}(I \cap I^\ast)=\{(v_8,v_1),(v_7,v_6)\}$ that do not get rotated when rotating from the left to the right tree are marked in red.} 
	\label{fig:concepts}
\end{figure}

A gap $(v_i,v_{i+1})$ in a tree $T$ is called \emph{empty} if the edge $(v_i,v_{i+1})$ is not part of $T$. 
Note that in the target tree $T^\ast$ from Lemma~\ref{convex_hull_star}, the gap  $(v_j,v_{j+1})$ is empty. 
Moreover, by setting $K^\ast=\emptyset$, a flip sequence as in Lemma \ref{convex_hull_star} can be used to rotate all edges to the convex hull such that all convex hull edges that point away from $v_n$ do not have to be rotated.

\subsection{The Sets $L$ and $R$}

\begin{proposition}\label{lr}
	There exists a flip sequence from $T_{\text{in}}$ to $T_{\text{tar}}$ that uses $2(n-1) -  max\big\{\lvert L \rvert, \lvert R \rvert\}$ rotations.
\end{proposition}

Without loss of generality, let $\lvert R \rvert > \lvert L \rvert$. For the sets $R$, we will (1) rotate all edges except the ones in $R$ to the convex hull such that every edge in $R$ has an empty gap to the left of its attachment point, (2) rotate all edges in pairs in $R$ to their target position, and (3) obtain the final tree. Steps (1) and (3) are described in Lemma \ref{cl}, Step (2) happens in Lemma \ref{lem:lrresolve}. An illustration of the strategy can be seen in Figure~\ref{LR}.

\begin{restatable}[$\star$]{lemma}{tohull}\label{cl}
	There exists a rotation sequence of length at most $n-1-\lvert R \rvert$ that rotates all edges except the ones in $R$ to the convex hull such that each edge is rotated into the gap in the convex hull that is left of its assigned vertex. In particular, for an edge $e=(v_i,v_r)$ with~$i<r$ that belongs to a pair in $R$, the gap $(v_{r-1},v_r)$ in the convex hull is empty. A symmetric argument holds for pairs in $L$.
\end{restatable}

\begin{figure}[ht]
	\centering
	\includegraphics[scale=0.55]{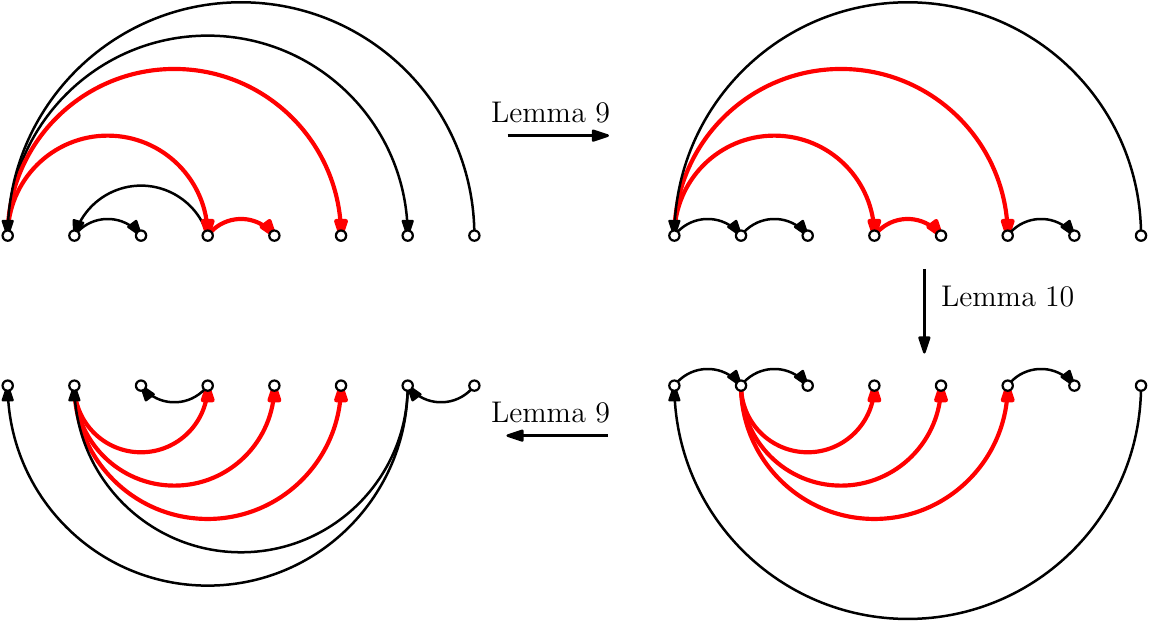}
	\caption{Steps of the proof of Proposition \ref{lr}. First, rotate all edges except the ones in R (marked in red) to the convex hull, then rotate the pairs in $R$. At last, rotate all remaining edges from the convex hull back in.} 
	\label{LR}
\end{figure}

Lemma \ref{cl} is an iterative application of Lemma \ref{convex_hull_star}. We apply the lemma  with $K^\ast=\emptyset$ to every subtree that results from cutting the tree along edges in $R$ (resp.\ $L$). Let $T_1$ and $T_2$ be the trees obtained from applying Lemma \ref{cl} to $T_{\text{in}}$ and $T_{\text{tar}}$.

\begin{restatable}[$\star$]{lemma}{lrresolve}\label{lem:lrresolve}
	There exists a rotation sequence from $T_1$ to $T_2$ of length at most $\lvert R\rvert$ (resp.~$\lvert L \rvert $).
\end{restatable}

The proof of Lemma \ref{lem:lrresolve} reduces to the notion of conflict graphs from \cite{bjerkevik2024flippingnoncrossingspanningtrees}. In Section~\ref{sec:upperbound}, left-attached and right attached pairs have been combined to above and below pairs, but we also get an acyclic conflict graph, if we pair them according to whether they are right-attached of left-attached.

\subsection{The Sets $D$ and $J$}

\begin{proposition}\label{dj}
	There exists a flip sequence from $T_{\text{in}}$ to $T_{\text{tar}}$ that uses $2(n-1) -  max\big\{\lvert D \rvert, \lvert J \rvert\}$ rotations.
\end{proposition}

As the first rotation of the rotation sequence we rotate $\rho_{T_{\text{in}},v_n}(v_1)$ into the edge $(v_1,v_n)$. From now on, we say that a \emph{gap is visible from above} if the gap is empty and shares a \emph{face} in the tree (more specifically a face of the drawing of $T_{\text{in}}$ combined with the convex hull) with $(v_1,v_n)$. The concept is illustrated in Figure \ref{fig:from_above}. Note that there is always a unique gap that is visible from above and that it has an edge $e_\ell$ that is attached to its left vertex and (maybe) an edge $e_r$ that is attached to its right vertex. $e_r$ may not exist if $g$ is the rightmost gap. Further, we show that $e_\ell$ and $e_r$ share a face with $(v_1,v_n)$.

\begin{lemma}
	$e_\ell$ and $e_r$ share a face with $(v_1,v_n)$.
\end{lemma}

\begin{proof}
	Assume $v_1$ and $v_n$ share a face with $k$ additional vertices. Then it also has to share a face with $k$ further edges in order for $T$ to be a tree. Those $k$ edges are exactly the edges that are attached to the $k$ vertices and, thus, include $e_\ell$ and $e_r$.
\end{proof}

\begin{figure}[ht]
	\centering
	\includegraphics[scale=0.55]{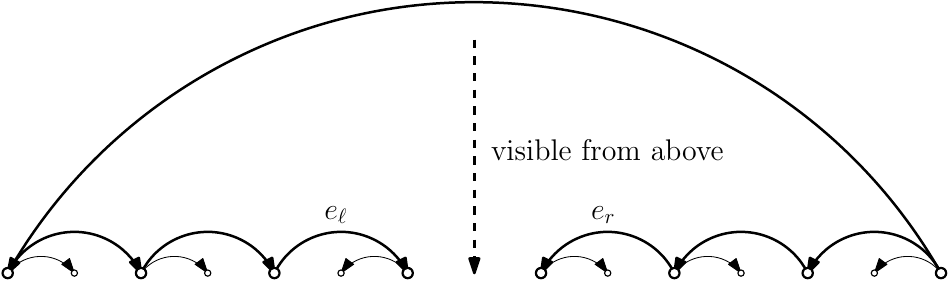}
	\caption{Gap that is visible from above} 
	\label{fig:from_above}
\end{figure}

\pagebreak

\begin{restatable}[$\star$]{lemma}{stars}\label{lem:stars}
	There exists a flip sequence of length at most $n-1$ that transforms $T_{\text{in}}$ into a tree $T^\ast$ that only contains the following three types of edges:
	\begin{itemize}
		\item[\textbf{(A)}] some convex hull edges attached to their right vertex,
		\item[\textbf{(B)}] all target edges of jump pairs of $(T_{\text{in}},T_{\text{tar}})$, and
		\item[\textbf{(C)}] some edges that are attached to their left vertex and have as a second vertex the joint vertex of a jump pair of $(T_{\text{in}},T_{\text{tar}})$ or the vertex $v_n$.
	\end{itemize}
\end{restatable}

\begin{figure}[ht]
	\centering
	\includegraphics[scale=0.55]{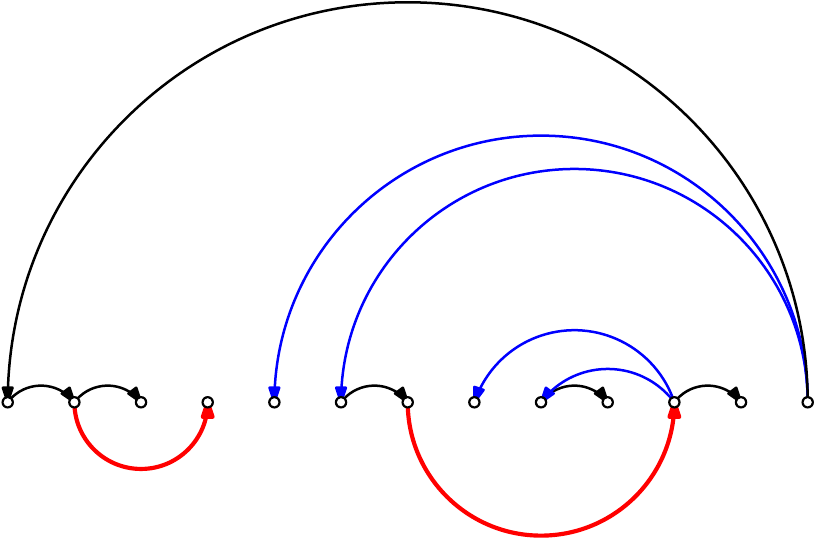}
	\caption{An intermediate tree $T^\ast$ as produced by Lemma~\ref{lem:stars}. The target edges from jump pairs are marked in red (and drawn below the vertices to emphasize that they are already at their target position). The edges of type $C$ from the set $K^\ast$ in Lemma~\ref{lem:djresolve} and are marked in blue.}
	\label{fig:Tstar}
\end{figure}

A tree $T^\ast$ as obtained from Lemma~\ref{lem:stars} is depicted in Figure~\ref{fig:Tstar}. The edges that are of special interest for the further application of Lemma~\ref{lem:djresolve} are highlighted with colors.

\begin{restatable}[$\star$]{lemma}{DJresolve}\label{lem:djresolve}
	We can transform $T_{\text{tar}}$ into $T^\ast$ as in Lemma \ref{lem:stars} in at most $n-1-\lvert J\rvert$ rotations.
\end{restatable}

To obtain Lemma \ref{lem:djresolve}, we apply Lemma \ref{convex_hull_star} to every face that results from cutting the tree along edges in pairs in~$J$.

Combining the rotation sequence from Lemma~\ref{lem:stars} (going from $T_{\text{in}}$ to $T^\ast$) and the reverse rotation sequence from Lemma~\ref{lem:djresolve} (going from $T^\ast$ to~$T_{\text{tar}}$), we obtain a rotation sequence between~$T_{\text{in}}$ and $T_{\text{tar}}$ of length at most $2(n-1) - \lvert J \rvert$, which proves Proposition~\ref{dj}.

\newpage

\section{Relations of Happy Edge Properties}\label{sec:happyedges}

In this section, we introduce a more refined distinction of happy edge properties.

\begin{definition}
	A graph reconfiguration problem, in which flips exchange edges,
	fulfills the
	\begin{itemize}
		\item \emph{(weak) happy edge property}, 
		if, from any initial graph $G_{\text{in}}$ to any target graph $G_{\text{tar}}$, there exists a shortest flip sequence that does not flip happy edges.
		\item \emph{strong happy edge property}, 
		if, from any initial graph $G_{\text{in}}$ to any target graph $G_{\text{tar}}$, any shortest flip sequence does not flip happy edges.
		\item \emph{perfect flip property}
		if, whenever we can perform a flip $f$ in $G_{\text{in}}$ such that the resulting graph $G_1$ has one edge more in common with $G_{\text{tar}}$ than $G_{\text{in}}$ has in common with $G_{\text{tar}}$, then there exists a shortest flip sequence from $G_{\text{in}}$ to $G_{\text{tar}}$ that has $G_1$ as its first intermediate graph (that is, the flip sequence starts with $f$ and $f$ is called a \emph{perfect flip}).
	\end{itemize}
\end{definition}

Further, we provide a framework, that allows us to compare the different properties.

\begin{restatable}[$\star$]{proposition}{hierarchy}
	\label{prop:hierarchy}
	Let $P$ be a graph reconfiguration problem where flips exchange one edge at a time.
	\begin{itemize}
		\item[(i)] If any flip sequence in $P$ that flips at least one happy edge can be shortened by at least two flips, then $P$ fulfills the perfect flip property.
		\item[(ii)] If for any flip in $P$ from $G$ to $(G\setminus\{e_1\})\cup\{e_2\}$, $e_1$ and $e_2$ can never be in the same configuration, then the reverse direction in (i) holds.
	\end{itemize}
\end{restatable}

We remark that all the introduced properties  
may or may not hold for 
certain graph reconfiguration problems. While the condition in Proposition~\ref{prop:hierarchy}(ii) is not fulfilled for trees, it does hold for example for triangulations.

\section{Happy Edges in Plane Spanning Trees on Convex Sets}\label{sec:happycompatible} 

In \cite{aichholzer2024reconfiguration}, the authors formulated the \emph{Weak Happy Edge Conjecture} for trees on convex point sets:
\begin{conjecture}[Conjecture~17 in \cite{aichholzer2024reconfiguration}] \label{conj:happyedgeconj}
For any two plane spanning trees $T_{\text{in}}$ and $T_{\text{tar}}$ on a convex point set, there is a shortest flip sequence from $T_{\text{in}}$ to $T_{\text{tar}}$ that does not flip happy~edges.
\end{conjecture}

Based on an example from~\cite{aichholzer2024reconfiguration} for a different context, we first observe that the perfect flip property holds neither for unrestricted flips nor for compatible flips on trees.

\begin{figure}[ht]
\centering
\includegraphics[scale=0.55]{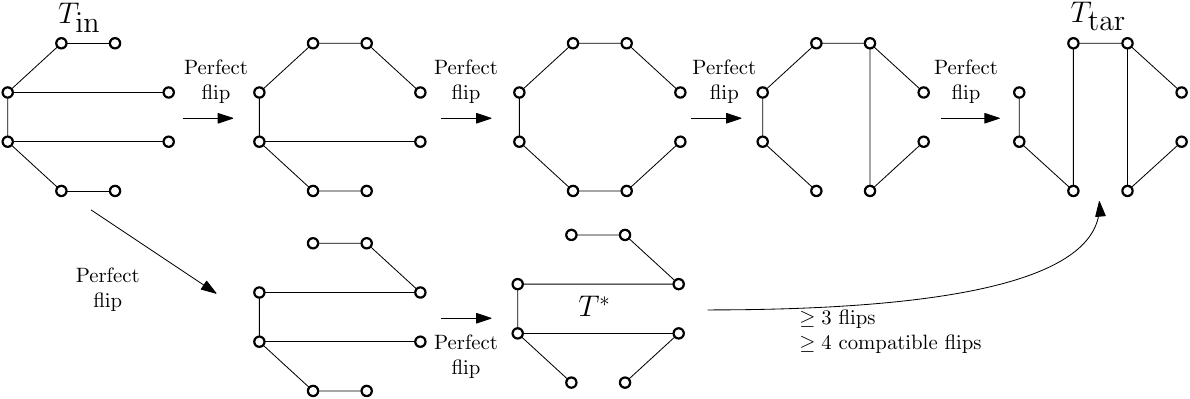}
\caption[Counterexample to the perfect flip property for trees.]{Counterexample to the perfect flip property based on \cite[Figure 7]{aichholzer2024reconfiguration}.
	The top shows the shortest flip sequence. The bottom sequence starts with two perfect flips, but reaches a point where no more perfect flips are possible. Since there are two edges in the tree $T^\ast$ that both cross two edges from the target tree, at least three additional flips or four additional compatible flips are needed.}
\label{perfect_flip}
\end{figure}

In contrast, we will show in this section that the strong happy edge property does hold for compatible flips on trees.
Our first proof ingredient is the following lemma, which is an extension of \cite[Proposition 18]{aichholzer2024reconfiguration} for compatible flips.

\begin{restatable}[$\star$]{lemma}{prehappy}
\label{prehappy}
Consider any point set $S$ and any two plane spanning trees $T_{\text{in}}$ and~$T_{\text{tar}}$ on~$S$ and any shortest compatible flip sequence from $T_{\text{in}}$ to $T_{\text{tar}}$. If some edge $e$ is removed and later added back, then some flip during that subsequence must add an edge $f$ that crosses~$e$.
\end{restatable}

In \cite{aichholzer2024reconfiguration}, a \emph{parking edge} is defined as an edge that appears in a flip sequence and that is not contained in~$T_{\text{in}} \cup T_{\text{tar}}$.
A second ingredient of our proof is Lemma~\ref{prehappy 3}, 
which verifies the compatible flip analogue of the following conjecture from~\cite{aichholzer2024reconfiguration}.

\begin{conjecture}[Parking Edge Conjecture, Conjecture 21 in \cite{aichholzer2024reconfiguration}] 
\label{prehappy2}
For any convex point set~$S$ and any two plane spanning trees~$T_{\text{in}}$ and~$T_{\text{tar}}$ on~$S$, there is a shortest flip sequence from~$T_{\text{in}}$ to~$T_{\text{tar}}$ that only uses parking edges from the boundary of the convex hull of $S$.
\end{conjecture}

We remark that in \cite[Claim 22]{aichholzer2024reconfiguration} it is shown that for unrestricted flips Conjecture~\ref{prehappy2} implies Conjecture~\ref{conj:happyedgeconj}.

\begin{restatable}[$\star$]{lemma}{prehappyy}
\label{prehappy 3}
For any convex point set $S$ and any two plane spanning trees~$T_{\text{in}}$ and~$T_{\text{tar}}$ on~$S$, there is a shortest compatible flip sequence from~$T_{\text{in}}$ to~$T_{\text{tar}}$ that only uses parking edges from the boundary of the convex hull of $S$.
\end{restatable}

\begin{figure}[ht]
\centering
\includegraphics[scale=0.55]{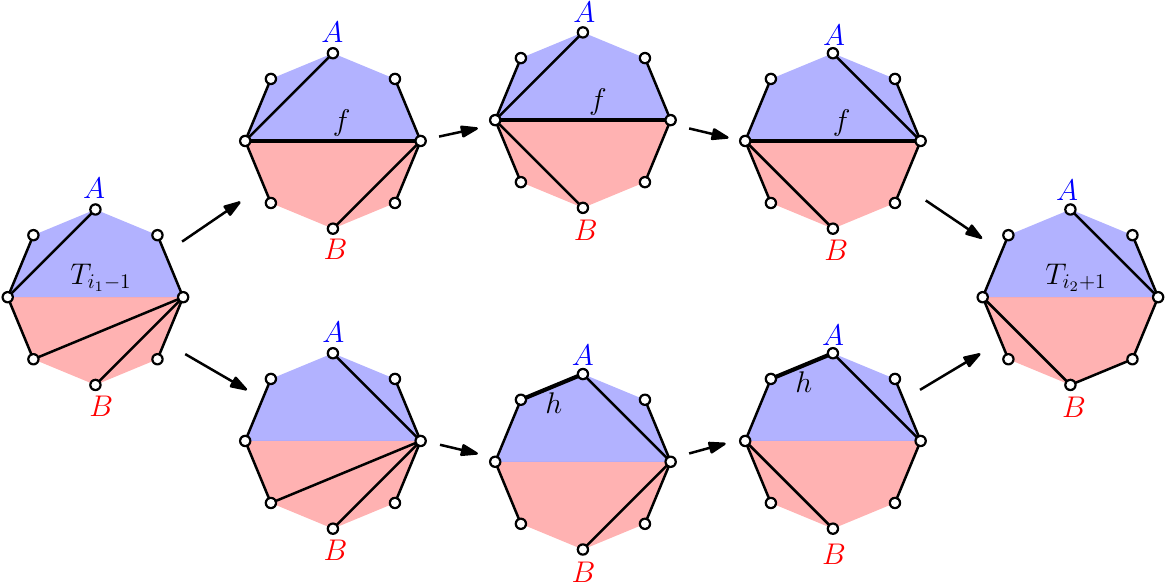}
\caption[Example of a reordering of steps.]{Reordering flips in Lemma~\ref{prehappy 3}. The top path is the part of the original flip sequence that contains $f$, the bottom path is the part of the reordered flip sequence that contains~$h$.}
\label{reorder}
\end{figure}

Figure~\ref{reorder} shows the intuition behind the proof: One by one, we replace a parking edge~$f$ that is not on the boundary of the convex hull with a parking edge $h$ on the convex hull boundary. To make this replacement possible, we change the order of the flips that happen while $f$ is part of the tree. The edge $f$ splits the convex point set into two sides, say $A$ and~$B$. Flips in either of the two sides can be executed independently from flips in the other side. Assume $f$ is added when flipping from the tree $T_{i_1-1}$ to $T_{i_1}$ and removed in the flip from $T_{i_2}$ to $T_{i_2+1}$. Exactly one of the sides, say side $B$, of $T_{i_1-1}$ entirely contains a path that connects the two endpoints of $f$. The flips in the other side, say side $A$, can be executed before $f$ gets added. Afterwards, we close a cycle in side $A$ by adding the convex hull parking edge $h$ and execute all the flips in side $B$. We conclude the new subsequence of flips by removing $h$ and obtain the tree $T_{i_2}$.

We now show the strong happy edge property for compatible flips on trees.

\begin{theorem}
\label{thm:happy}
For any convex point set $S$ and any two plane spanning trees $T_{\text{in}}$ and $T_{\text{tar}}$ on~$S$, 
any compatible flip sequence from $T_{\text{in}}$ to $T_{\text{tar}}$ that removes (and adds) a happy edge is at least one step longer than the shortest compatible flip sequence from $T_{\text{in}}$ to $T_{\text{tar}}$.
\end{theorem}

\begin{proof}
Let $T_{\text{in}} = T_0$, $T_1$,...,$T_k=T_{\text{tar}}$ be a compatible flip sequence that removes a happy edge~$e$ in a flip from~$T_i$ to~$T_{i+1}$. By Lemma~\ref{prehappy 3}, there exists a shortest compatible flip sequence from $T_{i+1}=T_{i+1}'$,$T_{i+2}'$,...,$T'_{\text{tar}}=T_{\text{tar}}$ that only uses parking edges from the convex hull boundary. 
Since $e$ crosses neither an edge of $T_{\text{in}} \cap T_{\text{tar}}$ nor one on the convex hull boundary,~$e$ crosses no edge in the flip sequence $T_{0}$,...,$T_{i}$,$T'_{i+1}$,...,$T_{\text{tar}}$. Hence, by Lemma~\ref{prehappy}, we can construct a shorter flip sequence that preserves $e$. 
\end{proof}

By Proposition~\ref{prop:hierarchy}(i) and the example in Figure~\ref{perfect_flip}, Theorem~\ref{thm:happy} is best possible in the sense that there are flip sequences that flip a happy edge and cannot be shortened by two flips.

When we restrict the flip even further by only allowing edge rotations instead of all compatible exchanges, the happy edge property no longer holds. Also for the more restricted edge slides, the happy edge property does not hold, see~\cite{AICHHOLZER2007155} for the general case, and~\cite{aichholzer2024reconfiguration} for convex point sets. 

Figure \ref{convex_case} shows a counterexample to the happy edge property for convex point sets with rotations. Assume the happy edge property would hold. Then the happy edge that passes through the middle of our convex polygon would split the tree into two proper subtrees. In each such subtree, the unhappy edge has a fixed position where it has to go. Since the initial position and the target position of the unhappy edge do not share any vertex, we need at least two rotations to get each unhappy edge from one position to the other. This gives us a minimum length of four for any flip sequence that preserves happy edges.

\begin{figure}[ht]
\centering
\includegraphics[scale=0.55]{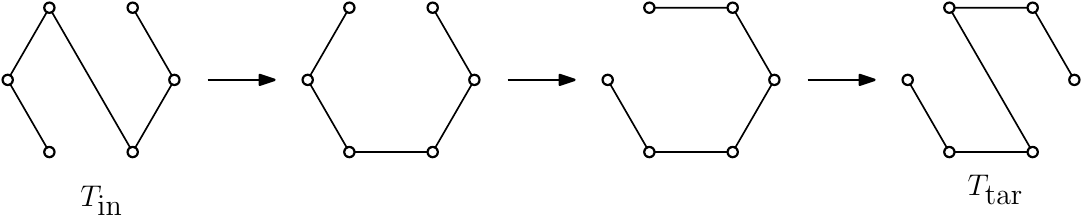}
\caption[Counterexample to the happy edge property for edge rotations in convex sets.]{Counterexample to the happy edge property for the reconfiguration of plane spanning trees in convex point sets using rotations.}
\label{convex_case}
\end{figure}

The flip sequence in Figure \ref{convex_case} only uses three rotations. Therefore, the happy edge property does not hold for the reconfiguration of plane spanning trees with rotations.

\section{Fixed-Parameter Tractable Algorithms}\label{sec:fpt}

Next, we show that the flip distance $k$ and the compatible flip distance $k_c$ are fixed-parameter tractable in $k$ (resp.\ $k_c$). The core ingredient for all proofs is the happy edge property for convex hull edges, which is shown to hold for unrestricted flips in Proposition 18 in \cite{aichholzer2024reconfiguration}. We obtain further improvements on the runtime built on the happy edge property for all edges and the parking edge property.

The upcoming lemma will describe a crucial step in the first two algorithms. Even if we never flip happy convex hull edges, the number of happy edges that are on the convex hull can still be arbitrarily large when compared to the number of unhappy edges. Therefore, we have to reduce that number in order to obtain a reasonable instance size that is bounded by a function in~$k$. We do this as follows: We mark all vertices that are incident to an unhappy edge that is in the initial configuration or the target configuration. All non-marked vertices are now only incident to convex hull edges. If we have a path between two marked vertices that only travels along convex hull edges incident to only non-marked vertices (except the first and/or last edge that start or end in a marked vertex), then we contract that path into one single edge (if both endpoints are marked) or a single vertex (if only one endpoint is marked). All the non-marked vertices get removed in the process. We call the new reduced instances $T_{\text{in,r}}$ and $T_{\text{tar,r}}$. The following Lemma states that in this contraction step no information about the length of the flip sequence is lost.

\begin{figure}[ht]
	\centering
	\includegraphics[scale=0.55]{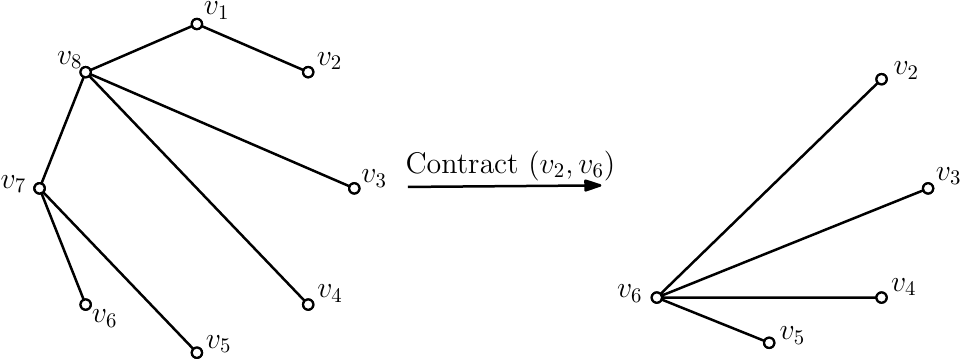}
	\caption[Contracting convex hull edges in a tree.]{An example of a tree before (left) and after (right) contracting the path from $v_2$ to $v_6$ along the convex hull.} 
	\label{fig:contraction}
\end{figure}

\begin{restatable}[$\star$]{lemma}{contract}
	\label{abc}
	The (compatible) flip distance between $T_{\text{in}}$ and $T_{\text{tar}}$ has the same length as the flip distance between $T_{\text{in,r}}$ and $T_{\text{tar,r}}$.
\end{restatable}

Now, we are ready to obtain our first FPT-algorithm. We remark that the notation $O^\ast$ hides terms that are polynomial in $n$.

\begin{theorem}\label{thm:FPT1}
	The flip distance $k$ for unrestricted flips is fixed-parameter tractable in $k$ with runtime $O^\ast((ck^3)^k)$, where $c=91.125$. The constant $c$ can be improved to $8$ if the happy edge property holds for unrestricted flips.
\end{theorem}

\begin{proof}
	For the unrestricted flip we define \emph{good happy edges}. 
	All happy edges on the convex hull are good happy edges.
	Further, a happy edge $e$ that does not lie on the convex hull splits the point set into two parts. 
	If one of the two parts contains only happy edges, then $e$ is a good happy edge as well.
	
	\begin{restatable}[$\star$]{claim}{goodhappy}\label{claim:good}
		No good happy edge is flipped in any shortest flip sequence.
	\end{restatable}
	
	Now consider the tree $N(T_{\text{in}})$ that is obtained from $T_{\text{in}}$ from exhaustively contracting happy leaves and paths of happy edges along the convex hull as in \cref{abc}. We consider the dual tree of $N(T_{\text{in}})$, whose vertices are the faces of $N(T_{\text{in}}) \cup \{\text{convex hull edges}\}$ except the outer face and which contains edge between two vertices if and only if their corresponding faces share an edge of $N(T_{\text{in}})$.
	
	\begin{restatable}[$\star$]{claim}{surround}
		\label{claim:surrounding}
		Let $P$ be a path in the dual tree that connects two faces containing distinct unhappy edges such that all the edges on traversed faces are happy edges. Then any shortest flip sequence either flips all those happy edges or none of them.
	\end{restatable}
	
	Let there be a total of $t$ unhappy edges. Then there are a total of $t-1$ paths $P_1$,...,$P_{t-1}$ that fit the form of Claim \ref{claim:surrounding}. Let $l_1$,...,$l_{t-1}$ denote their number of happy edges in some order. Any flip sequence of length $k$ gets to spend a total of at most $k-t$ flips on all these happy edges combined. We enumerate all $2^{t-1}<2^k$ subsets of $\{P_1,...,P_{t-1}\}$ and if their number of happy edges sums up to less than $\frac{k-t}{2}$, then we consider an instance that contains all the faces that contain unhappy edges and the faces along the paths in the subset.
	Then the total remaining tree contains at most $k$ unhappy edges from $T_{\text{in}}$ and $k$ unhappy edges from $T_{\text{tar}}$. That gives a total of at most $2k$ unhappy edges and a total of at most $4k$ vertices. Additionally, we have a total of less than $0.5k$ edges from faces with only happy edges. This gives another $0.5k$ vertices. So we have at most $4.5k$ vertices in our reduced forest. Brute forcing the optimal solution gives us a total of $4.5k-1$ possible edges to remove and a total of at most $\binom{4.5k}{2}$ positions to add edges for every flip. This gives us a total runtime of $O^\ast \left(2^k\left(4.5k\binom{4.5k}{2}\right)^k\right) = O^\ast((91.125k^3)^k)$ based on the number of flip sequences that we have to check.
	
	\begin{restatable}[$\star$]{claim}{FPT}
	\label{thm:fpt}
	If Conjecture~\ref{conj:happyedgeconj} is true, then the constant $c$ in Theorem~\ref{thm:FPT1} improves to $8$.
	\end{restatable}
	
	\begin{figure}[ht]
		\centering
		\includegraphics[scale=0.55]{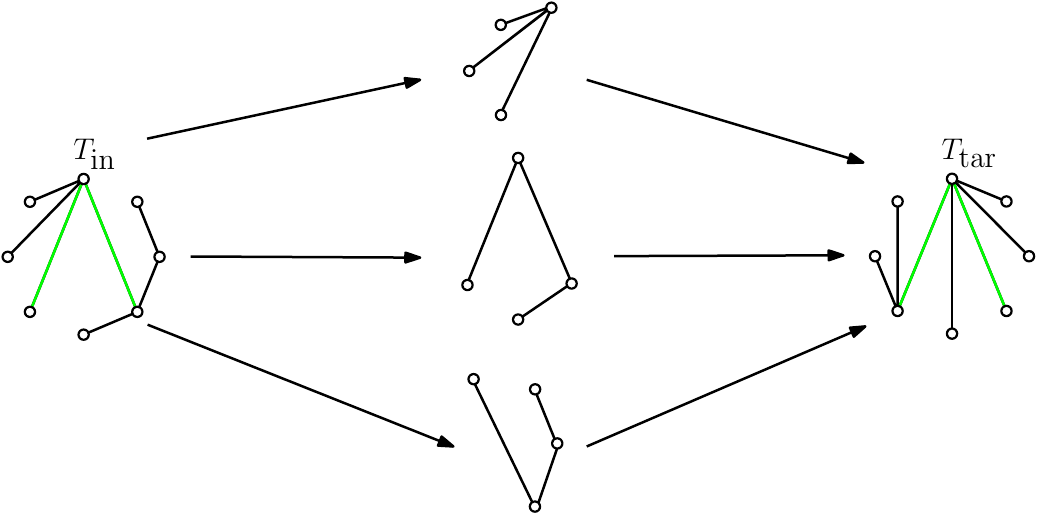}
		\caption[Splitting the original tree into smaller instances.]{We split the initial tree $T_{\text{in}}$ along its happy edges (colored in green) into three parts, from there we flip every part individually into its corresponding counterpart in $T_{\text{tar}}$.} 
		\label{split}
	\end{figure}
	
	Figure \ref{split} gives an intuitive idea, where the improvement comes from. We cut the point set along all happy diagonals such that the new instances have all their happy edges on  the convex hull.
\end{proof}

Next, we prove the existence of a faster FPT-algorithm for compatible flips that is based on the parking edge property. 

\begin{theorem} \label{thm:fpt2}
	The compatible flip distance $k_c$ is fixed-parameter tractable in $k_c$ with runtime $O^\ast(((2+\epsilon)k_c^2)^{k_c})$, for any fixed $\epsilon >0$.
\end{theorem}

\begin{proof}
	We divide the point set into components by cutting along happy edges and then find the shortest flip sequence for every remaining component individually, see Figure~\ref{split}. Observe that all the happy edges in these components are convex hull edges. If all edges in a component are happy, do not flip any edge in it. By Lemma~\ref{prehappy 3} there exists a shortest flip sequence that only uses edges from $T_{\text{in}}$, $T_{\text{tar}}$ or the convex hull.
	
	Now let $\epsilon > 0$ be fixed and $t = \big\lceil \frac{1}{\epsilon} \big\rceil$. Let~$C_1$,...,$C_m$ be all components with~$\leq t$ unhappy edges. Determine the respective flip distances~$k_j$ from~$T_{\text{in}}^{C_j}$ to~$T_{\text{tar}}^{C_j}$ via brute force and discard all the instances. For checking whether the remaining instances have a combined flip distance of $k'=k_c-k_1-...-k_m$ one has to check $((2+\epsilon)k'^2)^{k'}$ different flip sequences. There are $k'$ choices to remove an edge and every component with $l \geq n$ diagonals has $l+1 \leq l(1+\frac{1}{t})\leq l(1+\epsilon)$ holes in the convex hull where we can add edges.
\end{proof}

\section{Conclusion}\label{sec:conclusion}

In this work, we have considered the compatible flip graph and the rotation graph, by improving the upper bound on the respective diameter from $2n-o(n)$ to $\frac{5}{3}(n-1)$ and to~$\frac{7}{4}(n-1)$, respectively. The constructive upper bounds align with our findings about the happy edge property. On the one hand, we showed that any shortest compatible flip sequence preserves happy edges and that there is a shortest compatible flip sequence that only uses parking edges from the convex hull. The same holds for the compatible flip sequences that we used to obtain our upper bound. On the other hand, shortest rotation sequences may rotate happy edges. This may also happen in our upper bound construction. We further showed an immediate application of happy edge properties in trees, namely, that they can be used to derive fixed-parameter tractable algorithms. Open questions that are related to our work are:

\begin{enumerate}
\item Can we close the gap between upper and lower bounds for the flip distance, for each of the flip types? Do the diameters for different flip types differ or do they coincide?
\item Does the happy edge property still hold for flipping plane spanning trees if we drop one of the restrictions that either the point set is convex or the flips are all compatible?
\item What is the time complexity of finding shortest flip sequences for plane spanning trees? Again, this is interesting for each of the flip types.
\end{enumerate}

\bibliography{citation}

\newpage
\appendix

\section{Proof of Theorem \ref{upper}}\label{app:upper}

This section is devoted to the proof of Theorem~\ref{upper}.
Our goal is to modify the methods which were introduced and used in \cite{bjerkevik2024flippingnoncrossingspanningtrees} to show an upper bound of $\frac{5}{3} (n-1)$ for the diameter of the flip graph in order to improve the upper bound on the diameter of the compatible flip graph on a convex $n$-point set. 
As in~\cite{bjerkevik2024flippingnoncrossingspanningtrees}, the idea will be to pair edges~$(e,e')$ from the initial tree and the target tree such that the edge $e$ will either be flipped directly to the edge $e'$ performing one perfect flip or to flip $e$ to a convex hull edge and later flip the convex hull edge to $e'$ performing an \emph{imperfect flip}. The length of the flip sequence will then be $\#$perfect flips + $2$ $\#$imperfect flips. So we want to maximize the number of perfect flips.

We start by introducing some notation and results from~\cite{bjerkevik2024flippingnoncrossingspanningtrees}.
For two fixed trees $T_{\text{in}}$, $T_{\text{tar}}$ a \emph{linear representation} is introduced, which can be interpreted as cutting the cyclic order of the points on the convex hull between $v_1$ and $v_n$ and unfolding the circle into a horizontal line segment, called \emph{spine}. Each edge that was a straight-line chord can be thought of a semi-circle above (from $T_{\text{in}}$) or below (from $T_{\text{tar}}$) the spine. Recall that the linear order $v_1,...,v_n$ defines a cover relation between edges and vertices. An edge $(v_i,v_j)$ \emph{covers} a vertex $v_k$ if $i\leq k\leq j$. An edge $e$ \emph{covers} another edge $f$ if $e$ covers both vertices of~$f$. 
A \emph{gap}~$g_i$ is the segment along the spine with endpoints $v_i$ and $v_{i+1}$. 
For each gap $g$ and a set of non-crossing edges $E$, let $\rho_E(g)$ be the shortest edge of $E$ that covers $g$.

\begin{lemma}\cite[Lemma 3.1]{bjerkevik2024flippingnoncrossingspanningtrees}
	Let $E$ be a set of non-crossing edges on a linearly labeled point set $S$. 
	Then $\rho_E$ defines a bijection between the set of gaps and $E$ if and only if $E$ forms a tree on~$S$.
\end{lemma}


For a plane spanning tree $T$, the edges are labeled $e_i=\rho_T(g_i)$ and categorized into three types. For each $i\in [n-1]$ the edge $e_i = (u,v)$ is a
\begin{itemize}
	\item \emph{short edge} if $\{u,v\} = \{v_i,v_{i+1}\}$
	\item \emph{near edge} if $\lvert\{u,v\}\cap\{v_i,v_{i+1}\}\rvert = 1$
	\item \emph{wide edge} if $\{u,v\}\cap\{v_i,v_{i+1}\} = \emptyset$
\end{itemize}

The sets of short, near, and wide edges of $T_{\text{in}}$ are denoted by $S$, $N$, and $W$, respectively. 
Likewise, the sets of short, near, and wide edges of $T_{\text{tar}}$ are denoted by $S'$, $N'$, and $W'$, respectively.
Next, a set of pairs of edges $\mathcal{P}$ is defined, 
with $(e,e') \in \mathcal{P}$ if and only if $\rho_{T_{\text{in}}}^{-1}(e) = \rho_{T_{\text{tar}}}^{-1}(e')$.
The set $\mathcal{P}$ is further partitioned into

\begin{itemize}
	\item $\mathcal{P}_{=} = \{(e,e')\in\mathcal{P}\mid e=e'\}$,
	\item $\mathcal{P}_{N} = \{(e,e')\in\mathcal{P}\mid e\neq e', e\in N, e'\in N'\}$, and
	\item $\mathcal{P}_{R} = \mathcal{P}\setminus(\mathcal{P}_{=} \cup \mathcal{P}_{N})$
\end{itemize}

For pairs $(e_i,e_i') \in \mathcal{P}_N$ (and their corresponding gaps $g_i$), a more refined case distinction is made, depending on the relation between $e_i$ and $e_i'$.
\begin{itemize}
		\item \emph{above pairs}, where $e$ and $e'$ share a vertex and $e$ is longer than $e'$.
		\item \emph{below pairs}, where $e$ and $e'$ share a vertex and $e'$ is longer than $e$.
		\item \emph{crossing pairs}, where $e$ and $e'$ cross, or equivalently, $e$ and $e'$ share a different vertex with their gap.
\end{itemize}

Note that this case distinction is complete, since $e_i$ and $e_i'$ both cover their gap $g_i$ and, hence, the intervals of the spine which they cover overlap.
Let $A$, $B$, and $C$ denote the set of above, below and, crossing pairs, respectively. 

The next lemma follows from a series of results in \cite{bjerkevik2024flippingnoncrossingspanningtrees} where the authors produce a perfect flip sequence between two trees that only differ by either above pairs (A), below pairs (B), or crossing pairs (C) in the following way: 
By \cite[Lemma 3.2]{bjerkevik2024flippingnoncrossingspanningtrees} the sets A, B, and C are acyclic sets in the conflict graph. 
By \cite[Proposition~17 (arXiv version)]{bjerkevik2024flippingnoncrossingspanningtrees} there exists a perfect flip sequence for collections of pairs of edges that form an acyclic set in the conflict graph, given the rest of the tree consists only of convex hull edges. 
Since any edge pair in $A$ or $B$ is compatible, their flip sequences for above and below pairs are compatible, which directly implies the following statement.

\lemmaCompatibleAB*

In contrast, no direct flip between two edges of a pair in $C$ is compatible, since any such pair is crossing. 
Therefore, we have to come up with a different flip sequence in this case. Our main contribution here is the following Lemma.

\lemmaCompatibleC*

To prove Lemma~\ref{Lemma3}, we aim to change the way in which the edges from
$C$ are paired. 
Up to now, every edge in $T_{\text{in}}$ was flipped to end up in the position of its paired edge in $T_{\text{tar}}$. 
We will rearrange those target positions.

For this, we revisit the \emph{conflict graph} from~\cite{bjerkevik2024flippingnoncrossingspanningtrees}.
In a nutshell, this conflict graph $H$ has a set of near-near pairs (or, equivalently, the corresponding gaps) as vertices. 
Two vertices $(e_i,e_i')$ and $(e_j,e_j')$ in $H$ are connected with a directed edge from $(e_i,e_i')$ to $(e_j,e_j')$, 
if the direct flip $(e_j,e_j')$ is not possible while the edge $e_i$ is in the tree.
We use a scaled down version of $H_C$ of $H$ which has only the edge pairs of $C$ as vertices. 
In this case, a directed edge from $(e_i,e_i')$ to $(e_j,e_j')$ exists if the edge $e_j'$ crosses the edge $e_i$.

As part of \cite[Lemma 3.2]{bjerkevik2024flippingnoncrossingspanningtrees}, we get the following statement.

\begin{lemma}\label{Lemma4}
	$H_C$ is acyclic.
\end{lemma}

\begin{remark}\label{Lemma5}
	The proof of \cite[Lemma 3.2]{bjerkevik2024flippingnoncrossingspanningtrees} actually reveals an even stronger statement than Lemma~\ref{Lemma4}. 
	In this proof, the authors observe that for a directed edge from $(e_i,e_i')$ to $(e_j,e_j')$, in the conflict graph $H_C$,
	the curve formed by $(e_i,e_i')$ in the linear representation must at least partially lie above the according curve formed by $(e_j,e_j')$.
	They further show that there is at least one ``topmost'' pair $(e_i,e_i')$ in $C$ for which no other edge pair $(e_k,e_k')$ of $C$ lies partially above.
	Any such top-most pair $(e_i,e_i')$ is a root of $H_C$ with the properties that
	\begin{itemize}
		\item the corresponding edge $e_i$ of $T_{\text{in}}$ is not covered by any other edge~$e_j$ of $T_{\text{in}}$ with $(e_j,e_j') \in C$,
		\item the corresponding edge $e_i'$ of $T_{\text{tar}}$ covers no other edge~$e_j'$ of $T_{\text{tar}}$ with $(e_j,e_j')' \in C$, and 
		\item the edge $e_i'$ of $T_{\text{tar}}$ neither covers nor crosses any edge $e_j$ of $T_{\text{in}}$ with $(e_j,e_j') \in C\setminus{(e_i,e_i')}$.
	\end{itemize}
	Further, removing $(e_i,e_i')$ from $H_C$ yields a new conflict graph that, again, has a top-most root. %
\end{remark}

Every edge $e=(u,v)$ that does not lie on the boundary of the convex hull splits the vertex set into two parts, the points that appear in clockwise order along the convex hull boundary from $u$ to $v$ and those that appear in counterclockwise order. The vertices $u$ and $v$ are contained in both parts. 
The term \emph{face} with regard to a tree will refer to a inclusion-wise maximal subset of the point set such that it cannot be split into smaller parts by any edge of $T$. 
An edge is said to \emph{border} a face, if both its endpoints are contained in the vertex set of a face. 
Every convex hull edge borders exactly one face, every non-convex hull edge borders exactly two faces. 
Observe that every face contains exactly one pair of consecutive vertices of the convex hull that does not form an edge of the tree
(where except for $(v_n,v_1)$, all of those pairs of edges form gaps).

\begin{proof}[Proof of Lemma \ref{Lemma3}]
	Let $T_{\text{in}}$ and $T_{\text{tar}}$ consist of only common convex hull edges and pairs of edges in $C$, and consider the according conflict graph $H_C$ on $C$.
	By Lemma~\ref{Lemma4},  $H_C$ is acyclic and contains a root $(e_i,e_i') \in C$
	such that  
	$e_i'$ is not crossed by any of the initial edges~$e_k$ for any other pair $(e_k,e_k')$ 
	in $C$, and consequently no edge of our current tree. 
	Further, by Remark \ref{Lemma5}, we can choose $(e_i,e_i')$ such that $e_i$ is not covered by any edge in $T_{\text{in}}$ and $e_i'$ does not cover any edge in $T_{\text{in}}\Delta T_{\text{tar}}$.
	Let $(e_i,e_i')$ be chosen as above. The edge $e_i$ borders two faces of the tree $T$: 
	one face has the gap $g_i$ as non-edge 
	on the convex hull and the other face has $(v_n,v_1)$ as non-edge 
	on the convex hull. 
	We perform a flip that flips $e_i$ to $(v_n,v_1)$. 
	
	Now remove the pair~$(e_i,e_i')$ from the conflict graph. 
	Let $(e_j,e_j')$ be a pair of edges that is a 
	root of the reduced conflict graph $H_C$ that fulfills the properties of Remark~\ref{Lemma5} (for the reduced~$H_C$).
	We perform a flip that exchanges the edge~$e_j$ with the edge $e_i'$. 
	As $g_i$ and $g_j$ are each non-edges in different faces of $e_i'$,
	we do not close a cycle during the flip. 
	Further, by the choice of $e_i'$, the flip does not create a crossing and hence is compatible. 
	
	We remove the pair $(e_j,e_j')$ 
	from $H_C$ and iterate the process with $e_j'$ taking the role of $e_i'$. 
	We repeat the process until $H_C$ is empty. By the choice of $(e_j,e_j')$, adding $e_j'$ to the graph will never separate edges in pairs in $H_C$ from the face of $e_j'$ that contains $(v_1,v_n)$. Further, since $e_j'$ is not crossed by any edge in pairs in $H_C$, the next flip will be compatible. 
	Lastly, $g_j$ and $(v_1,v_n)$ (and therefore the convex hull non-edge from the previous step) will always lie in different faces of $e_j$, so we will not close a cycle in any flip.
	We conclude the flip sequence by performing a flip that exchanges $(v_n,v_1)$ for the remaining edge $e_i'$ that is yet to be added. 
	In total, we flipped $T_{\text{in}}$ into $T_{\text{tar}}$ in $\lvert C \rvert +1$ flips. For an illustration see Figure~\ref{shuffle}.
\end{proof}

\begin{figure}[ht]
	\centering
	\includegraphics[scale=0.75]{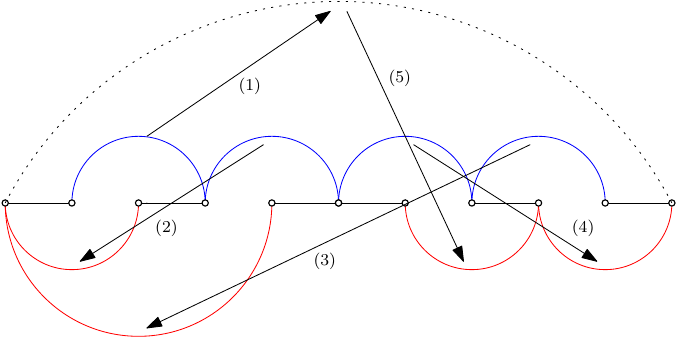}
	\caption{Example of a compatible flip sequence that fixes $4$ crossing pairs in $5$ flips.}
	\label{shuffle}
\end{figure}

Before proving Theorem~\ref{upper}, we need two more statements from \cite{bjerkevik2024flippingnoncrossingspanningtrees} (or, more exactly, from the arXiv version of the paper, since the conference version does not contain all proofs).

\begin{observation}[{\cite[Observation 28 (arXiv version)]{bjerkevik2024flippingnoncrossingspanningtrees}}] \label{Lemma6}
	Consider two plane spanning trees $T_{\text{in}}$, $T_{\text{tar}}$ on a convex point set $S$. If~$T_{\text{in}}\cup T_{\text{tar}}$ contain all convex hull edges of $S$, then there exists a compatible
	flip sequence of length $\lvert T_{\text{in}}\setminus T_{\text{tar}}\rvert$.
\end{observation}

We remark that the text of 
Observation~28 in the arXiv version of \cite{bjerkevik2024flippingnoncrossingspanningtrees} 
does not explicitly state that the flip sequence is compatible.	
However, the according flip sequence described in the text before that observation is clearly compatible, since every flip in it involves at least one convex hull edge.

\begin{lemma}[{\cite[Lemma 15 (arXiv version)]{bjerkevik2024flippingnoncrossingspanningtrees}}]\label{Lemma7}
	Let $T$ be a plane tree on linearly ordered points $v_1$,...,$v_n$. Let $k$ be the number of uncovered edges of $T$. Then $\lvert S\rvert \geq k$ and $\lvert W\rvert \leq \lvert S\rvert -k$.
\end{lemma}

Finally, we combine the above lemmata and observations to obtain the desired upper bound on the flip distance. 

\theoremCompatibleUpperBound*

\begin{proof}
	The proof follows the lines of the proof of \cite[Theorem~26 (arXiv version)]{bjerkevik2024flippingnoncrossingspanningtrees}.
	For the sake of self-containment, we include the proof in this appendix.
	
	We start with the case that $c=0$, where we distinguish two cases.
	If $T_{\text{in}}\cup T_{\text{tar}}$ covers all convex hull edges, then by Observation \ref{Lemma6}, there is a compatible flip sequence of length exactly $d$ from $T_{\text{in}}$ to $T_{\text{tar}}$. 
	Since every tree has at least two convex hull edges, $d=0$ implies $b\geq 2$ and $d=1$ implies $b\geq 1$.
	Therefore the flip distance in this case is $d \leq \frac{5}{3}d + \frac{2}{3}b -\frac{1}{3}$.
	
	Otherwise there exists a common gap on the convex hull. 
	We label the vertices in a way such that this gap corresponds to the edge $(v_n,v_1)$. 
	Then all pairs in $\mathcal{P}_{=}$ correspond to
	happy convex hull edges	and there are at least two uncovered edges.
	
	From Lemma \ref{Lemma7} we obtain $\lvert W \rvert + \lvert W'\rvert \leq \lvert S\rvert + \lvert S'\rvert -4$ and therefore $\lvert \mathcal{P}_R \rvert \leq \lvert W \rvert + \lvert W'\rvert + \lvert S\rvert + \lvert S'\rvert = 2\lvert S\rvert + 2\lvert S'\rvert -2b -4$. Additionally, $\lvert S\rvert + \lvert S'\rvert -2b$ gaps in $\lvert \mathcal{P}_R\rvert$ correspond to pairs with one short edge and only require one flip.
	
	By Lemma \ref{Lemma2} and Lemma \ref{Lemma3} we require at most $\frac{5}{3}\lvert \mathcal{P}_N \rvert+1$ flips to flip all pairs in $\mathcal{P}_N$. 
	The sequence is obtained by flipping all edges that belong to pairs of the smaller two components among $A$, $B$ and $C$ to the convex hull, then flipping the largest component perfectly (or with one additional flip in case $C$ is the largest) and then flipping all edges from the convex hull to their respective target edge.
	
	
	The flip sequence now consists of five parts:
	\begin{itemize}
		\item[(1)] Flip all edges $e_i$ from pairs in $\mathcal{P}_R$ to $g_i$
		\item[(2)] Flip all edges $e_i$ from the smaller two components among $A$, $B$ and $C$ to $g_i$
		\item[(3)] Flip all edges from the largest component among $A$, $B$ and $C$ perfectly (or with one additional flip in case $C$ is the largest)
		\item[(4)] Flip all edges from the smaller two components among $A$, $B$ and $C$ from $g_i$ to $e_i'$
		\item[(5)] Flip all edges $g_i$ from pairs in $\mathcal{P}_R$ from $g_i$ to $e_i'$
	\end{itemize}
	
	Let $d_1 = \lvert S \rvert + \lvert S'\rvert -4$ the number of pairs in $\mathcal{P}_R$ that require one flip, $d_2 = \lvert \mathcal{P}_{R} \rvert - d_1$ the number of pairs that require two flips, and $d_3 = \lvert \mathcal{P}_{N} \rvert$. Further, we have $d_2 \leq d_1 + 2b-4$.
	
	By putting everything together we obtain a flip distance of at most
	
	\begin{align*}
		d_1 + 2d_2 + \frac{5}{3}d_3+1 = d_1 + \frac{1}{3}d_2 +\frac{5}{3}(d_2+d_3)+1 \\ \leq d_1 + \frac{1}{3}(d_1+2b) +\frac{5}{3}(d_2+d_3)-\frac{1}{3} \\=\frac{4}{3} d_1 +\frac{5}{3}(d_2+d_3)+\frac{2}{3}b-\frac{1}{3} \\ \leq \frac{5}{3}d+\frac{2}{3}b-\frac{1}{3}
	\end{align*}
	
	Now if $c\neq0$, we cut the tree along happy non-convex-hull edges and perform the case~$c=0$ for each face $F$ individually. For that purpose, let $d(F)$ denote the number of edges in $T_{\text{in}} \setminus T_{\text{tar}}$ restricted to $F$ and $b(F)$ the number of common convex hull edges of $T_{\text{in}}$ and $T_{\text{tar}}$ of the induced subtrees on $F$. 
	Note that happy non-convex-hull edges become happy convex hull edges for the two components they border. We get the following upper bound for the flip distance:
	
	\begin{align*}
		\sum_F \frac{5}{3} d(F) + \frac{2}{3} b(F) - \frac{1}{3} = \frac{5}{3} d + \frac{2}{3}(b+2c) - \frac{1}{3}(c+1) = \frac{5}{3}d+\frac{2}{3}b+c-\frac{1}{3}
	\end{align*}
	
	The bound of $\frac{5}{3}(n-1)-\frac{1}{3}$ follows since $d+b+c=n-1$.
\end{proof}
	
	\section{Omitted Details from Section \ref{sec:rotations}}\label{app:rotations}
	
	\tohull*
	
	\begin{proof}
Without loss of generality we prove the claim for $R$, the argument for $L$ follows analogously. Every edge $e$ that belongs to a pair in $R$ cuts the convex point set $S$ into two \emph{sides}, the vertices that appear in clockwise order along the convex hull between the two vertices of $e$, and the vertices that appear in counterclockwise order. A \emph{face} is an inclusion-wise maximal subset of $S$ that is not cut into two parts by any edge that belongs to a pair in~$R$. 
Every face $F$ forms again a convex point set and we get an induced subtree $T[F]$ of $T$ on $F$ since every edge of $T$ has either both or none of its endpoints in $F$. 
All the edges that belong to pairs in $R$ and have both end vertices in $F$ are convex hull edges of $T[F]$.  

		For each such face $F$, we apply Lemma~\ref{convex_hull_star} to $F$, where we pick $v_n$ to be the rightmost vertex $v_r$ of~$F$ in the linear representation, and set the index $j$ to $r-1$.
Further, we set $K^\ast=\emptyset$ and let $I^\ast:=I^\ast[F]$ be the set of all indices of vertices of~$F$ except $r$.  %
		This application yields a rotation sequence of length $\lvert F \rvert-1-\lvert I[F] \cap I^\ast[F]\rvert$ where $I[F]$ is the set of all indices of vertices that are attached to convex hull edges of $T[F]$ and are directed away from $v_r$ in~$T[F]$. 

It remains to count how many rotations are performed in total. 
To this end, we split~$R$ into the set $R_{\text{diag}}$ of edges in $R$ that are diagonals and the set $R_{\text{CH}}$ of convex hull edges in~$R$. 
Observe that edges that belong to pairs in $R$ are all contained in $I[F] \cap I^\ast[F]$ for every face~$F$ that they bound. 
For calculating the number of rotations, we consider that 
		(1) there are $\lvert R_{\text{diag}} \rvert +1$ faces; 
		(2) all the vertices that belong to edges in pairs in $R_{\text{diag}}$ belong to exactly two faces; and 
		(3) pairs in $R_{\text{CH}}$ border exactly one face. 
This gives us a rotation sequence of the following total length.
		\begin{align*}
			\sum_{F~\text{face}} \lvert F \rvert-1-\lvert I[F] \cap I^\ast[F]\rvert = n+2\lvert R_{\text{diag}} \rvert-(\lvert R_{\text{diag}} \rvert +1) - \sum_{F~\text{face}} \lvert I[F] \cap I^\ast[F] \rvert \\ \leq n+2\lvert R_{\text{diag}} \rvert-(\lvert R_{\text{diag}} \rvert +1) - 2\lvert R_{\text{diag}} \rvert - \lvert R_{\text{CH}} \rvert = n - 1 - \lvert R \rvert
		\end{align*}
		The above estimate concludes the proof.
	\end{proof}

	\lrresolve*
	
	\begin{proof}
		
		Note that for every pair in $R$ (resp.\ $L$) the gap in the convex hull to the left (resp.\ right) of its assigned vertex is empty and both edges cover that gap. 
		We revisit the conflict graph for $R$ (resp.\ $L$). This time we use the full version of the conflict graph as defined in \cite{bjerkevik2024flippingnoncrossingspanningtrees} where the set of vertices is again given by the pairs of edges in $R$ (resp.\ $L$) and we add an edge from $(e_1,e_1')$ to $(e_2,e_2')$ if one of the three conditions holds:
		\begin{itemize}
			\item $e_1$ and $e_2'$ cross
			\item $e_2'$ covers $e_1$ and $e_1$ covers the gap left (resp.\ right) to the vertex assigned to $(e_2,e_2')$
			\item $e_1$ covers $e_2'$ and $e_2'$ covers the gap left (resp.\ right) to the vertex assigned to $(e_1,e_1')$
		\end{itemize}
		From \cite[Lemma 3.2]{bjerkevik2024flippingnoncrossingspanningtrees} we get that the set of all above pairs or the set of all below pairs form an acyclic subset of conflict graph, consequently the sets $RA$ and $RB$ (resp.\ $LA$ and $LB$) are both acyclic subsets of the conflict graph if regarded separately.
		
		It remains to show that the whole set $R = RA \cup RB$ (resp.\ $L=LA\cup LB$) combined form an acyclic set in the conflict graph.
		\begin{claim}\label{acyclic}
			The conflict graph on $R$ (resp.\ $L$) is acyclic.
		\end{claim}
		
		\begin{claimproof}[Proof of Claim \ref{acyclic}]
			This follows from the fact that the conflict graph only contains edges directed from pairs in $RA$ to pairs in $RB$ but not the other way around. Let $(e_1,e_1')$ be a right-attached below pair and $(e_2,e_2')$ be a left-attached above pair. Assume $e_1$ crosses $e_2'$ then also $e_1'$ crosses $e_2'$ a contradiction. If~$e_2'$ covers $e_1$ the vertex assigned to $(e_2,e_2')$ has to be to the right of the vertex assigned to $(e_1,e_1')$. Then also the assigned gap of $(e_2,e_2')$ is more right than all of $(e_1,e_1')$ therefore $e_1$ cannot cover the gap. Similarly, if $e_1$ covers $e_2'$, then the gap assigned to $(e_1,e_1')$ is too far to the right to be covered.
		\end{claimproof}
	
		Recall that each pair of edges in $R$ shares a common endpoint and hence any direct flip for such a pair is a rotation.
		Since the induced conflict graph of $R$ (resp.\ $L$) is acyclic, we obtain a sequence of direct rotations of all pairs in $R$ (resp.\ $L$) by always rotating a source of the current conflict graph; see~\cite[Proposition 17 (arXiv version)]{bjerkevik2024flippingnoncrossingspanningtrees}.
	\end{proof}
	
	\stars*
	
	\begin{proof}
		We use the following algorithmic approach: We consider the gap $g = (v_i,v_{i+1})$ that is visible from above and its incident edges $e_\ell$ and $e_r$ if they exist. Further, let $T$ be the tree obtained from $T_{\text{in}}$ by adding the edge $(v_1,v_n)$ via one rotation (in case it was not already in $T_{\text{in}}$). 
		Throughout the algorithm, we will perform rotations in $T$. 
		Further, we will modify $T$ by cutting along edges and contracting edges such that edges that need not be rotated anymore get removed. 
		This gives rise to the following case distinction; see Figure~\ref{fig:dj_cases} for an illustration:
		
		\begin{itemize}	
			\item[\textbf{(1)}] The edge $e_r$ belongs to a left attached pair: rotate $e_r$ into $g$
			\item[\textbf{(2)}] The edge $e_r$ belongs to a jump: Consider the edge $e_r':=\rho_{T_{\text{tar}},v_n}(v_{i+1}) = (v_j,v_{i+1})$, that is, the edge in $T_{\text{tar}}$ that is paired with $e_r$.
			\subitem\textbf{(2a)} All edges in $T$ that are attached to vertices between $v_{j+1}$ and $v_i$ are convex hull edges. Rotate $e_r$ into $e_r'$. Next, modify $T$ by \textbf{cutting} off all vertices that are now covered by $e_r'$ such that $e_r'$ becomes a convex hull edge in the resulting tree.
			\subitem\textbf{(2b)} There exists an edge that is attached to a vertex between $v_{j+1}$ and $v_i$ and is not a convex hull edge. 
				Among all such edges, let $f$ be the edge with the rightmost attachment point $v_k$. 
				We rotate $f$ into the edge $f'=(v_k,v_{i+1})$. 
				Then, we modify $T$ by \textbf{contracting} $f'$ and all edges below it into the vertex $v_{i+1}$. %
			\item[\textbf{(3)}] $g$ is the rightmost gap.
			\subitem\textbf{(3a)} If there exists some edge in $T$ that is not on the convex hull, take the rightmost vertex $v_k$ with a right attached edge $f$ attached to it and rotate $f$ to the edge $(v_k,v_n)$. Next, modify $T$ by contracting $f$ and all the edges below it into the vertex~$v_n$.
			\subitem\textbf{(3b)} If all other edges are on the convex hull, \textbf{terminate}.
		\end{itemize}
		
		\begin{figure}[ht]
			\centering
			\includegraphics[scale=0.55]{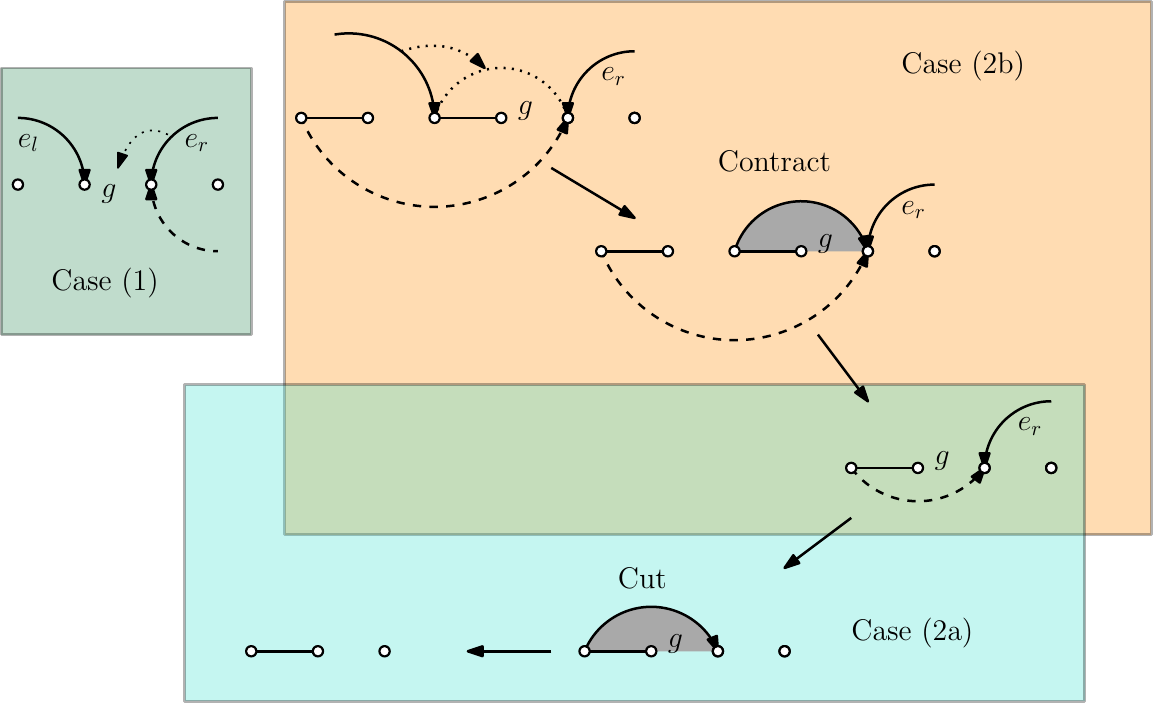}
			\caption{Case distinction based on the structure on the gap that is visible from above} 
			\label{fig:dj_cases}
		\end{figure}
	
	\begin{restatable}{claim}{correctness}\label{claim:stars}
		The algorithm terminates after at most $n-1$ iterations and the resulting tree after undoing all the cutting and contracting steps is $T^\ast$.
	\end{restatable}
	
	\begin{claimproof}[Proof of Claim \ref{claim:stars}]
		To see that we only execute one rotation per gap, observe that in Case (1) a gap is filled with an edge and never considered again in the remaining cases as the gap is removed from the graph by cutting or contracting.
		
		In Case (1), we only create edges of Type (A).
		
		In Case (2a), we only create edges of the form (A) and (B).
		
		Observe that in the Cases (2b) and (3a) the edge with the rightmost attachment point that is not on the convex hull is a right attached edge. Therefore, all edges between $v_k$ and $v_{i+1}$ are convex hull edges in the current state of $T$. Therefore, we do not get any edges other than the ones of Types (A) and (C) from those steps. Further, since in Case (2b), we observe that the edge $e_r$ is now attached to $v_k$, we look at a strictly shorter target edge for $e_r$ with strictly less right attached non-convex hull edges in the new state of the tree $T$. Therefore after finitely many occurrences of Case (2b) $e_r$ will be flipped into its target position.
	\end{claimproof}
	
	This concludes the proof of Lemma \ref{lem:stars}.
		\end{proof}
	
	\DJresolve*
	
	\begin{proof}
		Similarly to the proof of Lemma~\ref{cl}, we split $S$ into faces. However, this time, we use edges of pairs in $J$ for the splitting:
		In the following, a face is an inclusion-wise maximal subset of $S$ that is not cut into two parts by any edge in $T_{\text{tar}}$ (or, equivalently, in $T^\ast$) that belongs to a pair in $J$. 
		Every such face $F$ the points on the boundary of $F$ form a convex subset. 
		Since every edge of $T_{\text{tar}}$ has either both or none of its endpoints in $F$,
		we get an induced subtree $T_{\text{tar}}[F]$ of~$T_{\text{tar}}$ on~$F$.
		All the edges of $T_{\text{tar}}[F]$ that belong to pairs in $J$ and have both end vertices in $F$ are hence convex hull edges of $T_{\text{tar}}[F]$. 
		
		For each such face $F$, we apply Lemma \ref{convex_hull_star} to $F$, 
		where we pick $v_n$ to be the rightmost vertex $v_r$ of~$F$ in the linear representation and choose the index $j$ to be $r-1$. 
		Further, we set $K^\ast:=K^\ast[F]$ 
		to be the set of all indices of vertices that have an edge to $v_r$ in $T^\ast[F]$
		and let $I:= I^\ast[F]$ be the set of all indices of vertices that are attached to convex hull edges of $T^\ast[F]$ and are directed away from $v_r$ in $T^\ast[F]$. 
		This application yields a rotation sequence of length at most $\lvert F \rvert-1-\lvert I_{\text{tar}}[F] \cap I^\ast[F]\rvert$ where $I_{\text{tar}}[F]$ is the set of all indices of vertices that are attached to convex hull edges of $T_{\text{tar}}[F]$ and are directed away from $v_r$ in $T_{\text{tar}}[F]$. 

It remains to count how many rotations are performed in total. 
We split $J$ into the set $J_{\text{diag}}$ of edges in $J$ that are diagonals and the set $J_{\text{CH}}$ of convex hull edges in $J$. 
Note that all edges that belong to pairs in $J$ are contained in $I_{\text{tar}}[F] \cap I^\ast[F]$ for every face $F$ they belong to.
For calculating the number of rotations, we consider that 
		(1) there are $\lvert J_{\text{diag}} \rvert +1$ faces; 
		(2) all the vertices that belong to edges in pairs in $J_{\text{diag}}$ belong to two faces; and 
		(3) pairs in $J_{\text{CH}}$ only border one face. 
		This gives us a rotation sequence whose total length is bounded from above by the following expression\footnote{For an exact length, we would need to also substract the terms $\lvert K_{\text{tar}}[F] \cap K^\ast[F]\rvert$ for each face. We omit those, since we do not have a non-zero lower bound on them.}. %
		\begin{align*}
			\sum_{F~\text{face}} \lvert F \rvert-1-\lvert I_{\text{tar}}[F] \cap I^\ast[F]\rvert = n+2\lvert J_{\text{diag}} \rvert-(\lvert J_{\text{diag}} \rvert +1) - \sum_{F~\text{face}} \lvert I_{\text{tar}}[F] \cap I^\ast[F]\rvert \\ \leq n+2\lvert J_{\text{diag}} \rvert-(\lvert J_{\text{diag}} \rvert +1) - 2\lvert J_{\text{diag}} \rvert - \lvert J_{\text{CH}} \rvert = n - 1 - \lvert J \rvert
		\end{align*}
		The above estimate concludes the proof.
	\end{proof}
	
	\section{Omitted Details from Section \ref{sec:happyedges}}\label{app:a}
	
	\hierarchy*
	
	\begin{proof}
		\begin{figure}[ht]
			\centering
			\includegraphics[scale=0.55]{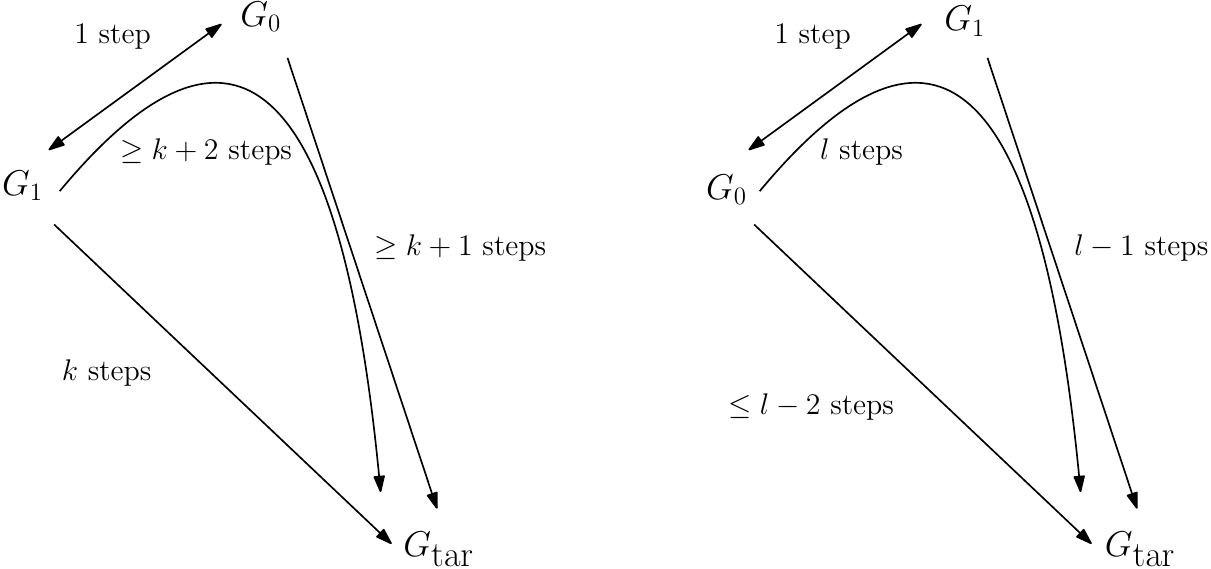}
			\caption{An improvement by two flips implies the perfect flip property (left) and the perfect flip property implies, under additional assumptions, an improvement by two flips (right)}
			\label{best_possible_2}
		\end{figure}
		
		\textbf{(i)}:
		Assume we can execute a perfect flip in $G_0$ to get to some $G_1$. Then any flip sequence from $G_1$ to $G_{\text{tar}}$ via $G_0$ has to be at least 2 flips longer than the shortest flip sequence from $G_1$ to $G_{\text{tar}}$ by assumption. So, if we need $k$ steps to get from $G_1$ to $G_{\text{tar}}$, then we need at least $k+2$ steps to get from $G_1$ to $G_{\text{tar}}$ via~$G_0$ and, therefore, at least $k+1$ steps from $G_0$ to $G_{\text{tar}}$. But going from $G_0$ to $G_1$ and then performing the sequence of length~$k$ from there gives us a flip sequence of length exactly~$k~+~1$ and, therefore, a shortest flip sequence. For an illustration see Figure~\ref{best_possible_2} (left).
		
		\textbf{(ii)}: Let $e_1$ be an edge that lies in both $G_{\text{in}}$, the initial graph, and $G_{\text{tar}}$, the target graph. Now let $G_{\text{in}} = G_0$, $G_1$,...,$G_k = G_{\text{tar}}$ be a flip sequence from $G_{\text{in}}$ to $G_{\text{tar}}$ that removes and later adds $e_1$. Further, assume without loss of generality that $e_1$ is removed in the very first flip from $G_0$ to~$G_1$ and that the edge $e_2$ is added in the same flip.
		
		Consider the flip sequence from $G_1$ to $G_{\text{tar}}$. By the perfect flip property, there exists a shortest flip sequence from $G_1$ to $G_{\text{tar}}$ via $G_0$, since $G_0$ and $G_{\text{tar}}$ have one edge more in common than $G_1$ and $G_{\text{tar}}$, since $e_1$ is in $G_{\text{tar}}$ and $e_2$ is not by assumption.
		Thus there is a shortest flip sequence from $G_1$ to $G_{\text{tar}}$ via $G_0$ that contains the shortest flip sequence from $G_0$ to $G_{\text{tar}}$ as a subsequence.
		
		This means that the shortest flip sequence from $G_0$ to $G_{\text{tar}}$ via $G_1$ has the same length as a flip sequence that first goes from $G_0$ to $G_1$, then back from $G_1$ to $G_0$, and then executes the shortest flip sequence from $G_0$ to $G_{\text{tar}}$. Since the first two steps are redundant, we can remove them and find a flip sequence that is at least two flips shorter then the given one. For an illustration see Figure \ref{best_possible_2} (right).
	\end{proof}
	
	\section{Omitted Details from Section \ref{sec:happycompatible}}\label{app:b}
	
	\prehappy*
	
	\begin{proof}
		
		The construction used in \cite{aichholzer2024reconfiguration} that proves Lemma \ref{prehappy} works as follows. Consider a flip sequence from $T_{\text{in}}$ to $T_{\text{tar}}$ and suppose that an edge $e$ is removed and later added back, and that no edge crossing $e$ is added during that subsequence. Let $T_0$,...,$T_k$ be the trees in this subsequence, where $T_0$ and $T_k$ contain $e$, but none of the intermediate trees do. Moreover,  none of the trees $T_i$, $1\leq i\leq k-1$, contains an edge that crosses $e$. We follow the lines in \cite{aichholzer2024reconfiguration} where they construct a shorter flip sequence from $T_0$ to $T_k$ using a normalization technique. For each $i$ with $1\leq i \leq k-1$ consider adding $e$ to $T_i$. For $i \neq 0,k$ this creates a cycle $\delta_i$. Let~$f_i$ be the first edge of $\delta_i$ that is removed during the flip sequence from $T_i$ to $T_k$. Define $N_i = T_i \cup \{e\} \setminus \{f_i\}$ and $N_0$ = $T_0$. This yields a new, shorter flip sequence $N_0$,...,$N_{k-1}$.
		
		Note that if the original flip sequence uses only compatible flips then the same holds for the normalized one. The reason is that any edge that gets added to some $N_i$ and that crosses the edge that gets removed from $N_{i-1}$ fits into two cases: It either crosses $e$, which is not possible by the assumption that no edge crossing $e$ gets added, or it crosses an edge that also lies in $T_{i-1}$ which contradicts the assumption that the original flip was compatible.
	\end{proof}
	
	\prehappyy*
	
	\begin{proof}
		We repeat the following construction for all parking edges that are introduced in a flip sequence and do not lie on the boundary of the convex hull. Further, we repeat the process for every appearance of such an edge separately, see Figure \ref{reorder}. 
		
		Let~$f=(v_a,v_b)$ be a parking edge that does not lie on the boundary of the convex hull. Let $T_{i_1}$ be the first tree to contain~$f$ and~$T_{i_2}$ be the last tree to contain~$f$ such that any tree between~$T_{i_1}$ and~$T_{i_2}$ contains~$f$. The trees~$T_0$,...,$T_{i_1-1}$ and $T_{i_2+1}$,...,$T_k$ stay as in the original sequence.
		
		We reorder and modify the other steps based on the following observation. The edge~$f$ separates the point set into two sides $A$ and $B$, the vertices that appear in clockwise direction from $v_a$ to $v_b$ and the vertices that appear in counterclockwise order. While $f$ is part of the tree, flips in the two sides can be executed independently from the other side. We remark that for this proof we characterize a flip by the two edges involved. Two valid flips are the same if they remove and add the same edges, regardless of the trees to which the edges belong to. Note that every spanning tree contains a unique path between any two vertices. If, in particular, we consider the path~$P$ from~$v_a$ to~$v_b$ in~$T_{i_1-1}$, we see that this path has to be contained entirely in one side induced by~$f$, w.l.o.g. in side~$B$. Otherwise, there would be an edge along the path that crosses~$f$, introducing a non-compatible flip from~$T_{i_1-1}$ to~$T_{i_1}$.
		
		If we add the edge~$f$, the cycle that gets closed is entirely contained in~$B$. Therefore, the edge that is removed lies in~$B$. Any flip that is executed from~$T_{i_1}$ to~$T_{i_2}$ that happens in~$A$ can already be executed in~$T_{i_1-1}$. Any cycle that previously involved the edge~$f$ now involves the path~$P$ instead. We order our flips in a way, such that all flips that happen in~$A$ are executed right away starting from~$T_{i_1-1}$. After executing all those steps, the side~$A$ already coincides with the side~$A$ in~$T_{i_2}$, except for the edge~$f$.
		
		Next we consider the part of~$T_{i_2}$ that lies in $A$. This part consists of $f$ and two trees~$T$ and $T'$ (possibly with an empty set of edges) in~$A$. Since~$T_{i_2}$ is plane there exist a unique pair of vertices that appear consecutively along the boundary of the convex hull and one belongs to~$T$ and the other to~$T'$. We continue our modified flip sequence by adding the convex hull edge~$h$ between the two vertices and remove the edge that would have been removed when adding~$f$.
		
		Now we execute all the flips in~$B$ from the original flip sequence. Any flip that formed a cycle that involves~$f$ now forms a cycle that involves~$h$ and paths in~$A$ connecting~$v_a$ and~$v_b$ to~$h$. In the last flip, we remove $h$ and add the edge that would have been added in the original sequence when removing~$f$. The resulting tree is~$T_{i_2}$.
		
		The new flip sequence is of the same length as the original one, but we replaced the non convex hull edge~$f$ with the convex hull edge~$h$. 
		\end{proof}
	
	\section{Omitted Details from Section \ref{sec:fpt}} \label{app:abc}
	
	\contract*
	
	\begin{proof}
		
		To show this, we observe that for any flip sequence from $T_{\text{in}}$ to $T_{\text{tar}}$ that preserves happy edges, we can obtain a flip sequence from $T_{\text{in,r}}$ to $T_{\text{tar,r}}$ of the same or shorter length and the other way around. For the backward direction, if we have a flip sequence from $T_{\text{in,r}}$ to $T_{\text{tar,r}}$, then we add and remove edges between vertices with the same labels to get from $T_{\text{in}}$ to $T_{\text{tar}}$, pretending the unmarked vertices don't exist.
		
		For the forward direction, assume we have a flip sequence from $T_{\text{in}}$ to $T_{\text{tar}}$, $T_{\text{in}}=T_0$, $T_1$,...,$T_k = T_{\text{tar}}$ that preserves happy edges. One by one, we contract a path $P~=$ $v_a$,$v_{a+1}$,...,$v_{b-1}$,$v_b$ (all indices taken modulo n), where all the $v_i$ are non-marked vertices that appear consecutively on the reduced convex hull, $v_a$ is the marked start and $v_b$ is the marked end, on the convex hull.
		
		We use a normalization technique, and distinguish two cases. \\[-2ex]
		
		\textbf{Case 1: Contract a path into an edge} See Figure~\ref{fig:contraction} for an example. Given a tree $T$ on the same vertex set as $T_{\text{in}}$ and~$T_{\text{tar}}$ that contains a path of edges that are happy with regard to $T_{\text{in}}$ and $T_{\text{tar}}$ such that further only the end vertices of the path are incident to non happy edges in $T_{\text{in}}$ or $T_{\text{tar}}$. We normalize with respect to the edge~$(v_a,v_b)$ to get the normalization $N(T)$ of $T$ by doing the following: Every edge that is not incident to any vertex on the path $P$ stays the same. All the edges along the path get removed and replaced by the edge~$(v_a,v_b)$. Any edge $(v_c,v_d)$ with $v_c \in \{v_{a+1},...,v_{b-1}\}$ and $v_d \notin \{v_a,v_b,v_{a+1},...,v_{b-1}\}$ gets replaced by the edge $(v_a,v_d)$. \\[-2ex]
		
		\textbf{Case 2: Contracting a happy leaf} Let $(v_a,v_b)$ be a happy leaf with respect to $T_{\text{in}}$ and~$T_{\text{tar}}$. Further, let $T$ be a tree on the same vertex set as $T_{\text{in}}$ and $T_{\text{tar}}^C$ that contains~$(v_a,v_b)$. Without loss of generality, let $v_b$ be the degree one vertex in $T_{\text{in}}$ and $T_{\text{tar}}$. We normalize as follows: Remove the edge $(v_a,v_b)$ from $T$ and replace any edge $(v_c,v_b)$ with the edge $(v_c,v_a)$.
		
		In both cases it holds that $N(T_{\text{in}}) = T_{\text{in,r}}$ and $N(T_{\text{tar}}) = T_{\text{tar,r}}$. Further, the normalizations $N(T_i)$ and $N(T_{i+1})$ of two consecutive trees $T_i$ and $T_{i+1}$ along the flip sequence either coincide (if the removed edge and the added edge are incident to vertices of the same contracted path or get slid along a happy leaf) or differ by a flip (all other choices of added and removed edges). After removing consecutive multiples, $N(T_0)$, $N(T_1)$,...,$N(T_k)$ is a sequence of trees that form a flip sequence from $T_{\text{in,r}}$ to $T_{\text{tar,r}}$. This is true, because, if we have a cycle in some $N(T_i)$, then this cycle must have appeared in $T_i$. The same holds for crossings.
		
		This concludes that we do not lose any information about the flip distance when contracting long paths along unmarked vertices along the convex hull.
			\end{proof}
			
	
	\goodhappy*
	
	\begin{claimproof}
		The claim clearly holds for convex hull edges \cite{aichholzer2024reconfiguration}. We apply the normalization from Lemma \ref{abc} to contract a happy leaf on the convex hull without changing the length of the shortest flip sequence. This contraction creates a new convex hull edge. Exhaustively contracting happy leaves, eventually every good happy edge will be a convex hull edge and therefore have the happy edge property holds.
	\end{claimproof}
	
	\surround*
	
	\begin{claimproof}
		Consider a path in the dual graph of $N(T_{\text{in}})$ between two faces that contain non-happy edges, such that all faces that appear at the interior of the path contain only happy edges. Assume that at least one edge $e$ is not flipped along the path. We cut $N(T_{\text{in}})$ and $N(T_{\text{tar}})$ along $e$ into two trees $N(T_{\text{in}})^1$ and $N(T_{\text{tar}})^1$, and $N(T_{\text{in}})^2$ and $N(T_{\text{tar}})^2$. Now, a shortest flip sequence from $N(T_{\text{in}})$ to $N(T_{\text{tar}})$ induces two shortest flip sequences from $N(T_{\text{in}})^1$ to $N(T_{\text{tar}})^1$ and from $N(T_{\text{in}})^2$ to $N(T_{\text{tar}})^2$. In both trees $e$ is a happy leaf that can be contracted without changing the length of the shortest flip sequence. Through the contraction, new happy leaves emerge. We contract happy leaves and happy paths along the convex hull exhaustively.
	\end{claimproof}

	\pagebreak
	\FPT*
	
	\begin{claimproof}
		Given two trees $T_{\text{in}}$ and $T_{\text{tar}}$. We cut our convex point set along the happy edges to obtain smaller subsets with formerly happy edges on their boundary and (possibly) unhappy edges in their interior. Afterwards, we discard all subsets that do not contain unhappy edges.
		
		Let $C$ be a subset that was not discarded in the process. We refer by $T_{\text{in}}^C$ to the graph induced by the initial tree on the point set of $C$ and analogously for $T_{\text{tar}}^C$.
		
		Finding the shortest flip sequence from $T_{\text{in}}$ to $T_{\text{tar}}$ is the same as finding the shortest independent flip sequences from $T_{\text{in}}^C$ to $T_{\text{tar}}^C$ for every single subset $C$ and then executing the respective flip sequences iteratively in $T_{\text{in}}$ to obtain $T_{\text{tar}}$. This is due to the assumption that Conjecture~\ref{conj:happyedgeconj} is true. The flips that can be executed on the points of one polygon $C$ without flipping a happy edge are all among the flips that can be executed in the reduced trees that correspond to $C$.

		By \cref{abc} we can again reduce the instance size by getting rid of redundant happy edges on the convex hull. We get the following runtime: The total number of marked vertices is at most $4k$ from having $2k$ unhappy edges, $k$ from $T_{\text{in}}$ and $k$ from $T_{\text{tar}}$, with $2$ endpoints each. The bound on the unhappy edges holds, since the flip distance is an upper bound on the number of unhappy edges. This leaves us with solving the instance for a forest with at most $4k$ vertices and $4k-1$ edges.
		
		The question, whether a flip sequence of length~$\leq k$ can be found, can be answered via brute force. At any given point, we have~$k$ choices to remove an unhappy edge and~$< \binom{4k}{2}$ choices to add an edge. Trying all sequences of length~$k$ gives us a total of~$<\left(k\binom{4k}{2}\right)^k$ sequences.
	\end{claimproof}
	
\end{document}